\documentclass[journal,12pt,onecolumn,draftclsnofoot]{IEEEtran}

\usepackage[noadjust]{cite}
\usepackage[pdftex]{graphicx}
\graphicspath{{../main_document/Figures/}} 
\DeclareGraphicsExtensions{.pdf,.jpeg,.png,.jpg}
\usepackage{url} 
\usepackage{mathtools}
\usepackage{amsfonts}
\usepackage{amssymb}
\usepackage{amsthm}
\usepackage{color,comment,dsfont}
\usepackage{subcaption}
\usepackage{xparse}

\usepackage{algpseudocode}
\usepackage{algorithm}
\usepackage{hyperref}

\usepackage [english]{babel}
\usepackage [autostyle, english = american]{csquotes}
\MakeOuterQuote{"}

\usepackage{amsmath}

\newcommand{\Z}{\mathbb{Z}}
\newcommand{\R}{\mathbb{R}}

\newcommand{\F}{\mathbb{F}}
\newcommand{\Ss}{\mathbb{S}}
\newcommand{\cR}{{\cal R}}
\newcommand{\cS}{{\cal S}}
\newcommand{\cA}{{\cal A}}
\DeclareMathOperator{\E}{\mathbb{E}}
\newcommand{\n}[1]{{\|\mathbf{#1}\|}}
\NewDocumentCommand\sqn{mg}{%
    \|\mathbf{#1}_{\IfNoValueTF{#2}{}{#2}}\|^2%
}
\newcommand{\FMOR}[1]{\textbf{\textcolor{cyan}{*** FIXME ORI: #1 ***}}}

\def\b{\mathbf}
\DeclareMathOperator*{\argmin}{arg\,min}
\DeclareMathOperator*{\argmax}{arg\,max}
\DeclarePairedDelimiter{\ceil}{\lceil}{\rceil}
\DeclarePairedDelimiter\floor{\lfloor}{\rfloor}
\newcommand\Algphase[1]{%
\vspace*{-.7\baselineskip}\Statex\hspace*{\dimexpr-\algorithmicindent-2pt\relax}\rule{0.45\textwidth}{0.4pt}%
\Statex\hspace*{-\algorithmicindent}\textbf{#1}%
\vspace*{-.7\baselineskip}\Statex\hspace*{\dimexpr-\algorithmicindent-2pt\relax}\rule{0.45\textwidth}{0.4pt}%
}
\DeclareMathOperator{\EX}{\mathbb{E}}

\newtheorem{theorem}{Theorem}
\newtheorem{lemma}{Lemma}

\newtheorem{corollary}{Corollary}

\newtheorem{definition}{Definition}

\newtheorem{remark}{Remark}

\begin{document}

\title{Compute-and-Forward in Large Relaying Systems: Limitations and Asymptotically Optimal Scheduling}



\author{\IEEEauthorblockN{Ori Shmuel, Asaf Cohen, Omer Gurewitz \thanks{This research was partially supported by European Unions Horizon 2020
Research and Innovation Program SUPERFLUIDITY, Grant Agreement
671566 and by MAFAT. Parts of this work appeared at the 2017 IEEE Information Theory Workshop (ITW) and at the 2018 IEEE Information Theory Workshop (ITW). The authors are from Ben-Gurion University of the Negev, Israel.
Email: \{shmuelor,coasaf,gurewitz\}@bgu.ac.il }}
}

\maketitle

\begin{abstract}

Compute and Forward (CF) is a coding scheme which enables receivers to decode linear combinations of simultaneously transmitted messages while exploiting the linear properties of lattice codes and the additive nature of a shared medium. The scheme was originally designed for relay networks, yet, it was found useful in other communication problems, such as MIMO communication. Works in the current literature assume a fixed number of transmitters and receivers in the system. However, following the increase in communication networks density, it is interesting to investigate the performance of CF when the number of transmitters is large.

In this work, we show that as the number of transmitters grows, CF becomes degenerated, in the sense that a relay prefers to decode only one (strongest) user instead of any other linear combination of the transmitted codewords, treating the other users as noise. Moreover, the system's sum-rate tends to zero as well. This makes scheduling necessary in order to maintain the superior abilities CF provides. We thus examine the problem of scheduling for CF. We start with insights on why good scheduling opportunities can be found. Then, we provide an asymptotically optimal, polynomial-time scheduling algorithm and analyze its performance. We conclude that with proper scheduling, CF is not merely non-degenerated, but, in fact, provides a gain for the system sum-rate, up to the optimal scaling law of $O(\log{\log{L}})$.

\end{abstract}

\begin{IEEEkeywords}
Relay Networks, User-Scheduling, Compute and Forward, lattice codes.
\end{IEEEkeywords}

\section{Introduction}
%
%
%
%

Compute and Forward (CF) \cite{nazer2011compute} is a coding scheme which enables receivers to decode linear combinations of transmitted messages, exploiting the broadcast nature of wireless networks. CF utilizes the shared medium and the fact that a receiver, which received multiple transmissions simultaneously, can treat them as a superposition of signals, and decode linear combinations of the transmitted messages. Specifically, with the use of lattice coding, the obtained signal after decoding can be considered as a linear combination of the transmitted messages. This is since for lattice codes, as a special case of linear codes, every linear combination of codewords is a codeword itself. However, since the wireless channel is subject to fading, the received signals are attenuated by real (and not integer) attenuations, hence the received linear combination is "noisy". In CF, the receiver (e.g., a relay) then seeks a set of integer coefficients, denoted by a vector $\b{a}$, to be as close as possible\footnote{One can define different criteria for the goodness of the approximation.} to the true channel coefficients and to serve as the coefficients for the linear combination it wishes to decode. The choice of coefficients, as well as the original channel gains, affect the resulting achievable rate.

Although CF was first introduced as a practical solution to mitigate users' interference in relays networks, the coding scheme and its properties were found to be an extremely useful tool in other types of communication channels and models. For example, linear MIMO receivers \cite{zhan2014integer}, decoding for the symmetric Gaussian K-user Interference Channel \cite{ordentlich2014approximate}, cloud radio access networks \cite{park2014fronthaul} and even in physical-layer security \cite{ling2014semantically}. In parallel, the CF scheme was extended to various settings which helped gain a deeper understanding of its abilities, advantages, and disadvantages. For example, MIMO CF  \cite{zhan2009mimo}, integration with interference alignment \cite{niesen2012degrees}, scheduling in cellular networks \cite{he2015collision}, MAC (Multiple Access Channel) \cite{zhu2017gaussian} and more \cite{wei2012compute,hong2013compute,lim2018joint}. 

In this work, we continue to shed light on the CF scheme and bring new results for a regime which, to the best of our knowledge, was not thoroughly discussed yet. That is, in contrast to the mentioned works, which assume that the number of simultaneously transmitting users is a fixed parameter of the system, in this work, we examine the effect of this number on the performance of the CF scheme. Specifically, we consider a general system of $L$ users and $M$ relays, and investigate the system sum-rate as $L$ increases. This asymptotic analysis, unlike the more common power asymptotic, takes a very realistic approach, which better fits current and future communication systems \cite{osseiran2014scenarios},\cite{dai2015non}. Such systems tend to have a dense topology, with low power and complexity devices, connected to a single or several access points. The ability of the CF scheme to handle simultaneous transmissions and use the users' interference constructively makes it a good candidate for such networks \cite{park2014fronthaul}. However, as we will show in the sequel, CF fails to perform well in such a regime, resulting in an asymptotically negligible rate. Yet, with proper scheduling and choices of coefficients, CF can retain its original benefits, and, in fact, \emph{distributively achieve the asymptotic upper bounds}.   


Numerous works discussed the capacity increase with the number of users, e.g., \cite{gupta2000capacity, gupta2003towards, gastpar2005capacity}, and several have also showed how to gain multi-user diversity with scheduling \cite{knopp1995information, qin2003exploiting, yoo2006optimality, shmuel2018performance}. However, we show that this is not necessarily the case with CF. In fact, a large number of simultaneous transmitters is not always a blessing when the number of relays is fixed, as when the former grows, the receiver will prefer to decode only the strongest user over all possible linear combinations. This will make the CF scheme degenerated, in the sense that a relay will choose a vector $\b{a}$ which is actually a unit vector (a line in the identity matrix), thus treating all other signals as noise. In other words, the linear combination chosen will be trivial. Furthermore, we show that as the number of transmitters grows, the scheme's \emph{sum-rate} goes to zero as well. Thus, one is forced to restrict the number of transmitting users, i.e., use scheduling, in order to maintain the superior abilities CF provides. We note that the idea of scheduling for CF networks was also considered in \cite{ramirez2015scheduling}. Therein, the authors showed by simulation that even a simple scheduling scheme can be useful. However, no performance guarantees or analysis for the optimal schedule were carried out.


Accordingly, we present a scheduling scheme for a large-scale relaying system that employs CF. The scheme is based on finding a group of users whose channel coefficients form a good coding opportunity for CF. Specifically, we present insight and analysis on what is considered as a good schedule for CF in terms of the channel's original coefficients and the coefficients of the approximated linear combination which is chosen to be decoded. We start our analysis from the point of view of a single relay, where we provide a lower bound on the expected system's sum-rate that can be achieved with our suggested scheduling scheme. Furthermore, we show that this bound is asymptotically optimal as the number of users grows. Then, we extend this scheduling policy for the case of a general system with multiple relays, for which we provide lower bounds on the expected system's sum-rate, prove the existence of good schedules for all relays simultaneously, and give heuristics for fast completion time.

\subsection{Paper Outline}
The paper is organized as follows. In Section \ref{Sec-Model and Problem Statement}, the system model is described. In Section \ref{sec-Main_results}, we present the main results, including both the results on the necessity of scheduling and results on the rates that can be achieved if scheduling is applied. Section \ref{Sec-CF_with_large_users} brings the analytical derivation for the probability of choosing a unit vector by a relay, as the number of users grows. This is the essence of the proof for the necessity of scheduling. Then, in section \ref{Sec-Scheduling_in_CF}, we present our user scheduling algorithm and its analysis for a system of a single relay. Section \ref{sec-Multiple Relays} extends the results to the case of multiple relays.

\section{System Model and Problem Statement}\label{Sec-Model and Problem Statement}

\subsection{Notational Conventions}
Throughout the paper, we will use boldface lowercase to refer to vectors, e.g., $\b{h} \in \R^L$, and boldface uppercase to refer to matrices, e.g., $\b{H} \in \R^{M \times L}$.  For a vector $\b{h}$, we write $\n{h}$ for its Euclidean norm, i.e. $\n{h} \triangleq \sqrt{\sum_{i}h_i^2}$. We denote by $\b{e}_i$ the unit vector with $1$ at the $i$-th entry and zero elsewhere.

\begin{figure}[t]
\centering
    \includegraphics[width=0.3\textwidth]{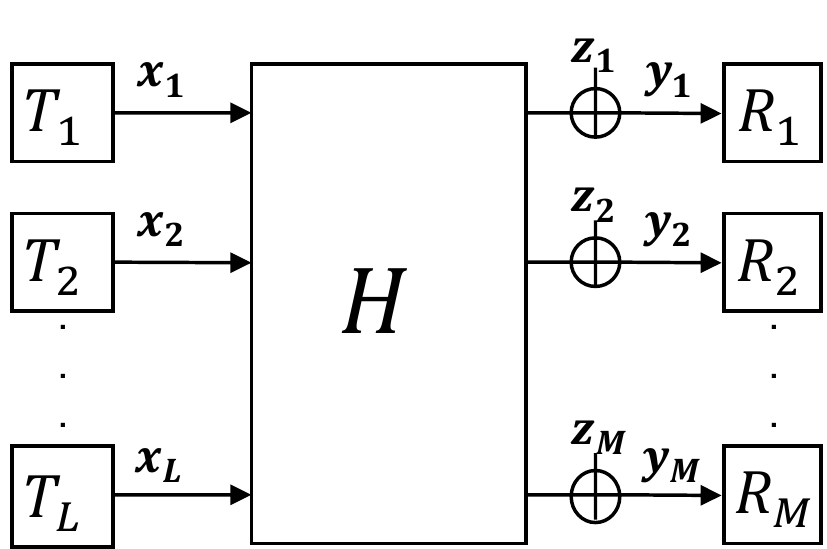}
\caption{Compute and Forward system model. $L$ transmitters communicate through a shared medium to $M$ relays.}
\label{fig-CF_System_model}
\end{figure}

\subsection{System Model}
Consider a multi-user multi-relay network, where $L$ transmitters (users) are communicating to a single destination $D$ via $M$ relays. All relays form a layer between the transmitters and the destination such that each transmitter can communicate with all the relays. Each transmitter has a length-$b$ message, assumed to be drawn $i.i.d.$ and with uniform probability over a prime size finite field, that is, $\b{w}_l \in \mathbb{F}_p^b, \ l=1,2,...,L$, where $\mathbb{F}_p$ denotes the finite field with a set of $p$ elements.

This message is then encoded by an encoder $\mathcal{E}_l:\mathbb{F}_p^b \rightarrow \R^n$, which maps the messages to length-$n$ real-valued codewords, $\b{x}_l=\mathcal{E}_l(\b{w}_l)$. Each codeword is subject to a power constraint, $\|\b{x}_l\|^2 \leq n\text{P}$. The message rate of each transmitter is defined by $R=\frac{b}{n}\log{p}$, measured in bits per channel use, and is equal for each transmitter\footnote{As our main contributions focus on the necessity of user scheduling for CF networks with a large number of users and the asymptotical optimal scheduling policies for such networks, we consider only equal transmissions rates. One can extend this formulation and the results for messages with different lengths and thus different rates as was done in \cite{nazer2011compute}.}. 
 
Accordingly, each relay $m\in\{1,...,M\}$ observes a noisy linear combination of the transmitted signals through the channel,
\begin{equation}\label{equ-the channel}
\b{y}_m=\sum_{l=1}^{L}h_{ml}\b{x}_l+\mathbf{z}_m \ \ \ \ m=1,2,...,M,
\end{equation}
where $h_{ml} \sim \mathcal{N}(0,1)$ are the real\footnote{We assume real channels to ease the analysis. The extension for complex channels is not critical for the results in this paper. As will be mentioned later, in the CF framework, complex channels change the achievable rates by a pre-log factor of 2.} channel coefficients and $\b{z}$ is an i.i.d., Gaussian noise, $\b{z} \sim \mathcal{N}(0,\b{I}^{n\times n})$. Let $\b{h}_m= (h_{m1},h_{m2},...,h_{mL})^T$ denote the vector of channel coefficients at relay $m$ and let $\b{H}=(\b{h}_1,\b{h}_2,...,\b{h}_M)^T$ be the channel matrix, i.e., the m-th' row of $\b{H}$ is $\b{h}_m^T$. We assume a memoryless block-fading channel model, i.e., the channel remains constant over each slot (block) period, and at the beginning of each slot independent realizations of $\b{H}$ are drawn; we assume that in each slot each relay knows its channel vector and that the channel vectors are independent of each other. This general model is illustrated in Figure \ref{fig-CF_System_model}.

\subsection{Compute and forward}
Nazer and Gastpar \cite{nazer2011compute} provided an achievable scheme which uses nested lattice codes for the computation of the linear equation of the transmitted signals over the channel \eqref{equ-the channel}. After receiving the noisy linear combination, each relay selects a scale coefficient $\alpha_m \in \R$, an integer coefficient vector $\b{a}_m=(a_{m1},a_{m2},...,a_{mL})^T \in \Z^L$, and attempts to decode the lattice point $\sum_{l=1}^La_{ml}\b{x}_l$ from $\alpha_m\b{y}_m$. Formally, the decoder has
\begin{equation}\label{equ-channel output at the decoder}
\begin{aligned}
	\alpha_m\b{y}_m&=\sum_{l=1}^{L}\alpha_mh_{ml}\b{x}_l+\alpha_m\b{z}_m =\sum_{l=1}^{L}a_{ml}\b{x}_l +\sum_{l=1}^{L}(\alpha_mh_{ml}-a_{ml})\b{x}_l+\alpha_m\b{z}_m.
\end{aligned}
\end{equation}

Due to the lattice algebraic structure, the relay decodes $\sum_{l=1}^La_{ml}\b{x}_l$ as a codeword while enduring the noise $\sum_{l=1}^{L}(\alpha_mh_{ml}-a_{ml})\b{x}_l+\alpha_m\b{z}_m$, namely, an \emph{effective noise} resulting from the true noise and the quantization error. 
The rate of the decoded codeword, i.e., the \emph{achievable rate}, defines a rate all transmitters must comply with to correctly decode the specific linear combination. The achievable rate and the optimal scale coefficient are given in the following two theorems.

\begin{theorem}[{\cite[Theorem 1]{nazer2011compute}}]\label{the-Computation rate}
For real-valued AWGN networks with channel coefficient vectors $\b{h}_m \in \mathbb{R}^L$ and coefficient vector $\b{a}_m \in \mathbb{Z}^L$, the following computation rate region is achievable:
\begin{equation}
\cR(\b{h}_m,\b{a}_m)= \max \limits_{\alpha_m \in \R} \frac{1}{2} \log^+ \left( \frac{\text{P}}{\alpha_m^2+\text{P}\|\alpha_m \b{h}_m-\b{a}_m\|^2} \right),
\end{equation}
\end{theorem}
where $\log^+(x)\triangleq \max \{\log(x),0\}$.
\begin{theorem}[{\cite[Theorem 2]{nazer2011compute}}]\label{the-Computation rate with MMSE}
The computation rate given in Theorem \ref{the-Computation rate} is uniquely maximized by choosing $\alpha_m$ to be the MMSE coefficient
\begin{equation}
\alpha_{MMSE}=\frac{\text{P}\b{h}_m^T\b{a}_m}{1+\text{P}\|\b{h}_m\|^2},
\end{equation}
which results in a computation rate region of
\begin{equation}\label{equ-Computation rate with MMSE}
\cR(\b{h}_m,\b{a}_m)= \frac{1}{2} \log^+ \left( \|\b{a}_m\|^2- \frac{\text{P}(\b{h}_m^T\b{a}_m)^2}{1+\text{P}\|\b{h}_m\|^2} \right)^{-1}.
\end{equation}
\end{theorem}
Note that the above theorems are for real channels and the rate expressions for the complex channel are twice the above \cite[Theorems 3 and 4]{nazer2011compute}. In addition, one should note that the coefficient vector $\b{a}_m$ must satisfy, 
\begin{equation}\label{equ-Search domain for the vector coefficients}
\|\b{a}_m\|^2 \leq 1+\text{P}\|\b{h}_m\|^2,
\end{equation}
so that computation rate in \eqref{equ-Computation rate with MMSE} would not be zero {\cite[Lemma 1]{nazer2011compute}}.

As mentioned, in order for relay $m$ to be able to decode a linear combination with coefficient vector $\b{a}_m$, the rate of the messages which have a non zero entry in $\b{a}_m$ must comply with the computation rate region \cite{nazer2011compute}. From transmitter's $l$ viewpoint, this means,
\begin{equation}\label{equ-Rates restrictions}
R_l<\min_{m:a_{ml} \neq 0} \cR(\b{h}_m,\b{a}_m).
\end{equation}
However, since we assume all transmission rates are equal, we have, $R=\min_{m} \cR(\b{h}_m,\b{a}_m)$ for all transmitters.

Recall that this work deals with the regime of large $L$ and a fixed $M$. Under this regime, the vectors (the channel and the coefficient vectors) in the rate expression in \eqref{equ-Computation rate with MMSE} have $L$ entries for any relay $m$. In the first part of this work, we analyze this rate and show that it tends to zero as $L$ grows, for all choices of $\b{a}_m$. This will lead to the conclusion that a restriction on the number of actually simultaneously transmitting users is required. That is, one must schedule users for transmission to achieve a non-zero rate. On the other hand, we will show that intelligent scheduling is indeed possible, resulting in strictly positive rates, and, in fact, scale-optimal performance.

\subsection{Compute and forward - Complete Decoding}\label{sec-Compute and forward - Complete Decoding}

The previous subsection provided the restriction on the achievable rate such that a single relay will be able to decode a single linear combination. For the complete decoding of all messages, the destination must acquire at least $L$ independent linear combinations on the $L$ messages and solve the linear system. The collecting process of the linear combinations can be done by multiple relays, multiple transmission slots or by a combination of the two.
Thus, the ratio between $L$ and $M$ has a crucial effect on the time the destination will be able to decode all messages, i.e., the completion time. That is, it may happen that a single transmission of the messages would not suffice to decode all messages. Clearly, for $L>>M$, we expect to have several transmission slots until complete decoding is possible. In each such slot, the users repeat their transmissions at the same rates and the destination collects $M$ coefficient vectors and their corresponding decoded codewords. We define by $\b{A}_n=(\b{a}_1(1),...,\b{a}_{M}(1),...,\b{a}_1(n),...,\b{a}_{M}(n))^T$ the decoding matrix which is formed by the coefficient vectors collected in $n$ transmissions as it's rows. 

\begin{remark}[Number of transmission slots]
We note here that the number of transmission slots may be reduced if a relay decodes several equations from the same transmission \cite{nazer2012successive, ordentlich2014approximate}. However, this ability restricts the achievable rate to the lowest equation decoded. Moreover, the results in the first part of this work show that without proper scheduling the achievable rate goes to zero for any coefficient vector $\b{a}$, in particular, the one which maximizes $\cR(\b{h},\b{a})$; thus, any other linear equation will result in an even lower rate.      
\end{remark}
We point out the following observations. First, note that since the channel changes between the slots, the transmission rates should also be constrained by the weakest time slot. Moreover, the throughput should be normalized by the number of slots used since a transmitter transmits the same message in each slot (until successful decoding). Second, since there is no coordination between the relays, the $M$ coefficient vectors decoded in each slot by the relays may be linearly dependent, either within a slot or across slots; thus, the number of Degrees of Freedom (DoF) may be less than $M$, increasing the number of slots needed for complete decoding.
Accordingly, in order to have a successful decoding of all messages at transmission slot $N$, we require that $rank(\b{A}_N)=L$ and that 
\begin{equation}\label{equ-rate restriction for successful decoding}
R<\min_{n=1,...,N}\min_{m} \cR(\b{h}_m(n),\b{a}_m(n)).
\end{equation}
This assures that the linear system of equations can be solved and that each equation has the same rate and was decoded correctly by the relays. Note that $N$ is a r.v. since it depends implicitly on the channel vectors. Note also that the all linear operation are taken modulo $p$ \cite[Theorem 7 and Remark 9]{nazer2011compute}

We may now express the system sum-rate when successful decoding is possible, which is the sum of rates of the $L$ original messages (in our case, each of rate $R$) divided by the number of slots until complete decoding. 
\begin{definition}[The system sum-rate]
The system sum-rate is,
\begin{equation}\label{equ- system sum-rate definition}
C_{SR} \triangleq \frac{R  L}{N}, 
\end{equation}
were $R$ is subject to \eqref{equ-rate restriction for successful decoding} and $N$ is the first slot for which $rank(\b{A}_N)=L$.
\end{definition}

\subsection{Problem statement}\label{sec-Problem statement}

The performance of such a system depends on the proper choice of the coefficient vectors by the relays in each transmission slot. This choice rules the transmission rate of the messages and the completion time for successful decoding. Essentially, one would want to maximize the rate in each transmission slot, on one hand, i.e., find coefficient vectors that are the best fit for the real channel vectors. On the other hand, one would want to minimize the time by making sure these vectors contribute additional DoF in each slot. As mentioned, there is no coordination between the relays when it comes to the selection of these $\b{a}$ vectors; therefore, for most of this work, we focus on the considerations for this selection from the perspective of a single relay, i.e., we consider a single relay model for which $M=1$.  By this choice, one can remove, in general, the minimum constraint on the transmission rates as in \eqref{equ-Rates restrictions}, since there is only a single achievable rate, and analyze this system more easily. Hence, from this point and on, we omit the index $m$ and consider a single relay model. 
In Section \ref{sec-Multiple Relays} we return to the case of multiple relays and elaborate on its implications.

Our goal is to maximize the system's sum-rate given in \eqref{equ- system sum-rate definition}. Thus, since the rate and the completion time are both a function of the coefficient vectors selections made by the relay in each slot, we may write this maximization as follows,


 \begin{equation}\label{equ-Problem statment}
\max_{\substack{\b{A}_N \in \Z^{N\times L}} } \frac{R(\b{A}_N)  L}{N(\b{A}_N)},
\end{equation}
where now $\b{A}_N$ includes the coefficients vectors of the single relay. We note that one may suggest a sub-optimal solution for this problem by letting the relay choose, in each slot, the coefficient vector $\b{a}$ that maximizes the achievable rate regardless the completion time $N$ \footnote{Since the slots are independent and there is no consideration of the rank of $\b{A}$, we may write the maximization on an arbitrary slot and omit the time index.}. That is,
\begin{equation}\label{equ-Optimal a vector definition}
\textbf{a}^{max}=\argmax_{\textbf{a}\in\mathbb{Z}^L \backslash \{\textbf{0}\}}\frac{1}{2}\log^+\left(\|\textbf{a}\|^2-\frac{\text{P}(\mathbf{h}^T\mathbf{a})^2}{1+\text{P}\|\mathbf{h}\|^2} \right)^{-1}.
\end{equation}
Accordingly, the destination would have to wait until collecting enough independent linear combinations to attain a full rank.

The problem of finding the maximizing $\b{a}$ can be done by exhaustive search for small values of $L$. However, as $L$ grows, the problem becomes prohibitively complex quickly. In fact, it becomes a special case of the lattice reduction problem, which has been proved to be NP-complete. This can be seen if we write the maximization problem of \eqref{equ-Optimal a vector definition} as an equivalent minimization problem \cite{sahraei2014compute}:
\begin{equation}\label{equ-Optimal a vector in quadratic form}
\textbf{a}^{max}=\argmin_{\textbf{a}\in\mathbb{Z}^L \backslash \{\textbf{0}\}} f(\textbf{a})=\textbf{a}^T\textbf{G}\textbf{a},
\end{equation}
where $\textbf{G}=(1+\text{P}\|\b{h}\|^2)\b{I}-\text{P}\b{h}\b{h}^T$. $\b{G}$ can be regarded as the Gram matrix of a certain lattice and $\b{a}$ will be the shortest basis vector and the one which minimizes $f$. This problem is also known as the \emph{shortest lattice vector} problem (SLV), which has known approximation algorithms due to its hardness \cite{dadush2011enumerative, alekhnovich2005hardness}. The most notable of them is the LLL algorithm \cite{lenstra1982factoring, gama2008finding} which has an exponential approximation factor that grows with the size of the dimension. However, for special lattices, efficient algorithms exist \cite{conway2013sphere}. In \cite{sahraei2014compute}, a polynomial complexity algorithm was introduced for the special case of finding the maximizing coefficient vector in CF.

In what follows, we show that regardless of complexity considerations, when $L$ grows the optimal solution for \eqref{equ-Problem statment} tends to zero, and in order to promise a positive system sum-rate one must require that not all users transmit simultaneously. This extends the maximization problem in another dimension - where scheduling is also allowed. While the problem is even more complicated with scheduling, in this paper we show that simple scheduling algorithms exist, they result in non-negligible rates, hence overcome the limitations of a large $L$, and, moreover, we present a low-complexity, asymptotically optimal algorithm for joint scheduling for CF which achieves the best possible scaling law.

\section{Main Results}\label{sec-Main_results}
The main results of this work can be divided into two main threads. The first thread shows that the optimal solution for \eqref{equ-Problem statment} results in negligible rate if indeed all users transmit together. Hence, one must use \emph{user scheduling} to avoid this pitfall. The second thread shows that under scheduling, not only strictly positive rates are possible, but actually such schedules can be easy to implement, distributed, and asymptotically achieve the optimal scaling laws for the expected system sum-rate. We summarize these in the following subsections. The proofs are given in Sections \ref{Sec-CF_with_large_users}, \ref{Sec-Scheduling_in_CF} and \ref{sec-Multiple Relays}.

\subsection{Necessity of Scheduling}
When the number of relays is fixed and the number of transmitters is large, scheduling a small number of users is necessary for the CF scheme not to degenerate. Specifically, we first show that letting all users transmit will cause a relay, with a probability that goes to one with the number of transmitters, to choose a vector $\b{e}_i$ as the coefficient vector which maximizes the achievable rate. This will result in decoding only the strongest user and not a non-trivial linear combination.  

\begin{theorem}\label{the-Probability for having a unit vector as the maximaizer}
Under the CF scheme, the probability that a non-trivial vector $\b{a}$ will be the coefficient vector which maximize the achievable rate $\mathcal{R}(\b{h},\b{a})$, i.e., minimize $f(\b{a})$ in \eqref{equ-Optimal a vector in quadratic form}, compared to the best unit vector $\b{e}_i$, i.e., $i=\argmin_i f(\b{e}_i)$, is upper bounded by
\begin{equation}\label{equ-Probability of unit vector as minimizer of f}
P_r( f(\b{a}) \leq \min_i{f(\b{e}_i)}) \leq 1- I_{\Phi(\b{a})}\left(\frac{1}{2},\frac{L-1}{2}\right),
\end{equation}
where $I_x(a,b)$ is the CDF of the Beta distribution with parameters $a$ and $b$, and $\Phi(\b{a})=1-\frac{1}{\|\b{a}\|^2}$. 
\end{theorem} 
Note that $\frac{1}{2} \leq \Phi(\b{a}) \leq 1$ for any $\b{a}$ which is not a unit vector.

The main consequence of Theorem \ref{the-Probability for having a unit vector as the maximaizer} is the following. 

\begin{corollary}\label{cor-Probability for having a unit vector as the maximaizer goes to one}
 As the number of simultaneously transmitting users grows, the probability that a non-trivial $\b{a}$ will be the maximizer for the achievable rate goes to zero. Specifically, 
 \begin{equation}\label{equ-Probability of choosing a unit vector goes to one}
P_r( f(\b{a}) \leq \min_i{f(\b{e}_i)}) \leq e^{-LE_1(L)},
\end{equation}
where $\b{a}$ is any integer vector that is \emph{\textbf{not}} a unit vector, $\b{e}_i$ is a unit vector and  $E_1(L)=(1-\frac{3}{L})\log{\|\b{a}\|}$.
\end{corollary}

Corollary \ref{cor-Probability for having a unit vector as the maximaizer goes to one} clarifies that for every power $\text{P}$, as the number of users grows, the probability of having a non-trivial vector $\b{a}$ as the maximizer of the achievable rate tends to 0. Note that the assumption of $L>3$, which arises naturally from this paper's regime, along with the fact that $\|\b{a}\|\geq2$, guarantees that $E_1(L)$ is positive. Figure \ref{fig-Probability_for_having_unit_vector} depicts the probability in \eqref{equ-Probability of unit vector as minimizer of f}, it's upper bound given in \eqref{equ-Probability of choosing a unit vector goes to one} and simulation results. From the analytic results as well as the simulations on the rate of decay, one can deduce that even for relatively small values of simultaneously transmitting users ($L>20$), a relay will prefer to choose a unit vector. Also, one can observe from the results and from the analytic bound that as the norm of $\b{a}$ grows, the rate of decay increases. This faster decay reflects the increased penalty of approximating a real vector using an integer-valued vector.


Note that Corollary \ref{cor-Probability for having a unit vector as the maximaizer goes to one} refers to the probability that a non-trivial \emph{fixed} $\b{a}$ will be the maximizer of $f$. Next, we wish to explore this probability for any possible $\b{a}$ \emph{that is not a unit vector}, yet satisfies $\|\b{a}\|^2< 1+\text{P}\|\b{h}\|^2$. Define by $P_r(\b{e})$ the probability that a relay picked any unit vector as the coefficient vector, and by $P_r(\overline{\b{e}})$  the probability that a non-trivial vector was chosen.

\begin{theorem}\label{the-Probability for having a unit vector as the maximaizer over all other vectors}
Under the CF scheme, the probability that any other non-trivial coefficient vector $\b{a}$ will be chosen to maximize the achievable rate $\mathcal{R}(\b{h},\b{a})$ compared with any unit vector $\b{e}_i$, as the number of simultaneously transmitting users grows, is zero. That is, 
\begin{equation}
\lim_{L \rightarrow \infty} P_r(\overline{\b{e}})= 0.
\end{equation}
\end{theorem}

\begin{figure}[t]
\centering
    \includegraphics[width=0.5\textwidth]{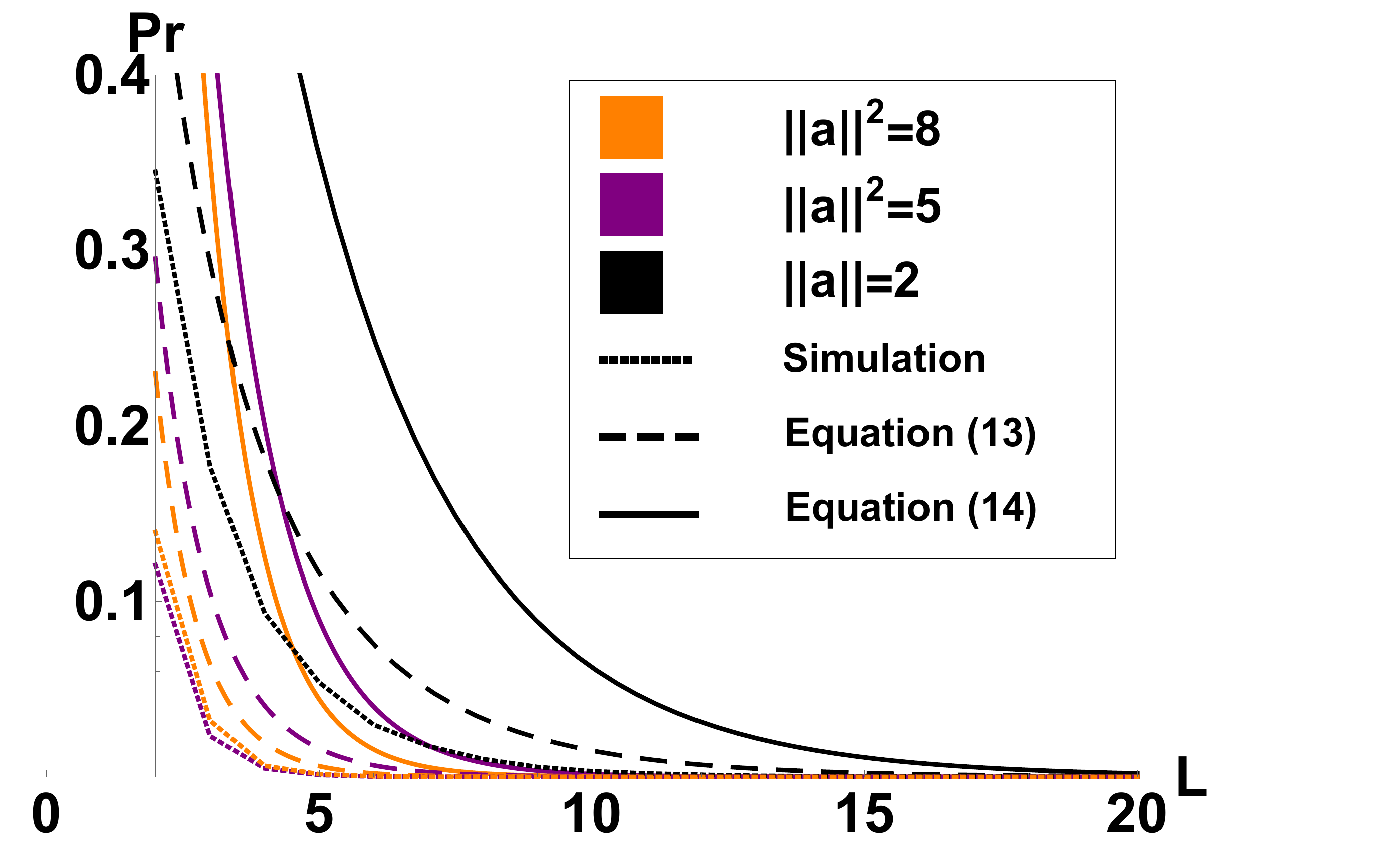}
\caption{The upper bounds on the probability of not having a unit vector as the minimizer of $f$ ,given in \eqref{equ-Probability of choosing a unit vector goes to one} (solid lines) and \eqref{equ-Probability of unit vector as minimizer of f} (dashed lines)  compared to simulation results (dotted lines) for various values of $\|\b{a}\|^2$. The $x$ axis is the number of simultaneously transmitting users.}
\label{fig-Probability_for_having_unit_vector}
\end{figure}

The main consequences of Theorem \ref{the-Probability for having a unit vector as the maximaizer over all other vectors} are the following results, that apply for the general model with $M$ relays. The results describe the behavior of the achievable rate of any relay and the system's sum-rate at the limit of a large number of simultaneously transmitting users.
\begin{theorem}\label{the-Achievable is going to zero}
As $L$ grows, the achievable rate of a relay converges to zero in probability, that is, 
\begin{equation}
\lim_{L \rightarrow \infty}  P_r(\cR(\b{h},\b{a})>\epsilon) = 0,
\end{equation}
for all $\epsilon>0$.
\end{theorem}
\begin{corollary}\label{the-Sum rate is going to zero}
As $L$ grows, the sum-rate of CF for a general system with $M$ relays converges to zero  in probability, that is,
\begin{equation}
\lim_{L \rightarrow \infty}  P_r\left(\frac{R  L}{ N}>\epsilon\right) = 0,
\end{equation}
for all $\epsilon>0$, where $R<\min\limits_{n=1,...,N}\min_{m} \cR(\b{h}_m(n),\b{a}_m(n))$.
\end{corollary}

Theorem \ref{the-Achievable is going to zero} and corollary \ref{the-Sum rate is going to zero} show that the sum-rate of all decoded linear combinations tends to zero as the number of simultaneously transmitting users grows. Thus, applying the CF coding scheme in its vanilla settings, where all transmitters transmit and the relays decode linear combinations of all transmitted messages, is futile. A possible course of action, while still applying CF, is to restrict the number of simultaneously transmitting users in each transmission slot. In the sequel, we show that by performing this restriction, not only a strictly positive sum-rate is achievable, smart scheduling can provide an overall gain to the system's sum-rate. 

\subsection{Scheduling in CF}\label{subsec-Scheduling in CF}


In this part of the work, we will suggest a scheduling policy, and prove its asymptotic optimality by meeting a global upper bound on the system's sum-rate. The scheduling policy is based on setting the number of scheduled users in each slot properly, as well as identifying the specific subset and a specific set of coefficient vectors from which the relay chooses the linear combination for the scheduled users. We note that for the case of a general system, with $M>1$ relays, our suggested scheduling policy should be further adapted to the restrictions which arise from the presence of multiple relays; we discuss this in Section \ref{sec-Multiple Relays}.

We first present the scheduling problem under the model of $L$ users and a single relay. We first assume that there is a scheduler (e.g., the relay itself) that handles the scheduling process. In Section \ref{sec-Distributed Scheduling} we discuss the implementation of the scheduling process in a distributed manner, i.e., without having a centralized entity. We assume that in each slot $i$ a subset of $k$ users is chosen by the scheduler. This subset is denoted by $\cS_k(i)$. The value of $k$ is fixed and will be dealt with in the sequel. Accordingly, the total number of subsets is ${L \choose k}$, each having a channel vector which we denote by $\b{h}(\cS_k(i))$, and a corresponding coefficient vector $\b{a}(\cS_k(i))$. In addition, we denote by $\b{h}_L(i)$ the channel vector of all users, to create a distinction with the channel vector of the scheduled users. As mentioned above, the scheduling problem consists of two highly connected optimization problems. The first can be viewed as finding the proper subset of users $\cS_k(i)$, and the second is finding the proper $\b{a}(\cS_k(i))$ for the $\b{h}(\cS_k(i))$ of the selected users, to maximize the system's sum-rate. Let us denote by $\cS_k^N=\left(\cS_k(1),\cS_k(2),...,\cS_k(N)\right)$ the sequence of all subsets of users which where scheduled until complete decoding. In addition, denote by $\cA_k^N=\left(\b{a}(\cS_k(1)),\b{a}(\cS_k(2)),...,\b{a}(\cS_k(N))\right)$ the coefficients matrix with the coefficient vectors in all slots. Note that since $\cA_k^N\in \Z^{N\times k}$, yet decoding is preformed over dimension $L$, we must map $\cA_k^N$ to $\b{A}_N$ by setting each row of $\b{A}_N$ to have the entries of $\b{a}(\cS_k(i))$ for the scheduled users and $0$ otherwise. We denote this map as $(\cdot)_\Uparrow^L:\b{a}(\cS_k(i)) \mapsto \b{a}(i)$. For example, let $\cS_3(1)=\{1,2,5\}$ be the 3 users that where scheduled, with coefficient vector $\b{a}(\cS_3(1))=(1,1,2)$ in the first slot. Then, for $L=6$, we have $(\b{a}(\cS_3(1)))_\Uparrow^6=\b{a}(1)=(1,1,0,0,2,0)$. Furthermore, it will be useful to define the set of vectors of length $L$, which have $k$ non-zero values taken from a certain set $\Omega$ as $\Ss_{L,k}^\Omega$. For example, for $\Omega=\{1\}$, the set $\Ss_{L,2}^{\{1\}}$ is all binary vectors of length $L$ with exactly two ones.

We thus modify the maximization in \eqref{equ-Problem statment} to include our scheduling problem,
\begin{equation}\label{equ-Problem statment with scheduling}
\max_{\substack{\cS_k^N, \ \cA_k^N \in \Z^{N\times k}} } \frac{R(\cS_k^N,\cA_k^N)  L}{N((\cA_k^N)_\Uparrow^L)}.
\end{equation}

In this work, we provide a polynomial-time (in both $k$ and $L$) scheduling algorithm, described in Algorithm \ref{algo-scheduling algorithm for all transmission}, Section \ref{Sec-Scheduling_in_CF}. Algorithm \ref{algo-scheduling algorithm for all transmission} finds the asymptotically (with $L$) optimal schedule for the maximization problem in \eqref{equ-Problem statment with scheduling} for all transmission slots as well as the coefficient vectors which ensure complete decoding. Its asymptotic guarantees are bellow.

\begin{theorem}\label{the-Expected achievable rate of scheduling algorithm lower bound}
\textit{The asymptotic expected achievable rate Algorithm \ref{algo-scheduling algorithm for all transmission} achieves in each slot, for $L\geq4$,
is lower bounded by the following,}
\begin{equation*}
	\EX\left[\cR_{ach}^{sch}\right] \geq  \frac{1}{2}  \log^+ \left( k \left(1-\frac{\text{P}ku^4}{(u+\delta)^2(1+\text{P}ku^2)}(1-o(1)) \right)\right)^{-1},
\end{equation*}
\textit{where $u=\sqrt{2\ln{\frac{\delta\sqrt{L}}{\sqrt{2\pi}}}}-\delta$, $\delta=\frac{1}{\ln{L}}$ and $o(1) \rightarrow 0$ as $L \rightarrow \infty$.}
\textit{Thus, the expected achievable rate for Algorithm \ref{algo-scheduling algorithm for all transmission} in each slot scales at least as $O(\frac{1}{4}\log{\log{L}})$. }
\end{theorem}
The values for $u$ and $\delta$ relate to the characteristics of the scheduled users' channel, and determine the search domain for the optimal schedule, as will be explained in the sequel. They are chosen to provide asymptotically optimal result. The suggested scheduling scheme, however, is independent of $u$ and $\delta$ and works well even for moderate number of transmitters, as Figure \ref{fig-Optimal schedule CF comparisons} depicts. In addition, although it is left as a parameter, the number of scheduled users $k$ in each slot, under our scheduling paradigm, should be $O(\log{L})$; we refer to this choice in the sequel as well.
We thus have the following corollary.
\begin{corollary}\label{cor-expected system sum-rate is lower bound using scheduling}
The expected system's sum-rate that Algorithm \ref{algo-scheduling algorithm for all transmission} can achieve for a single relay system is lower bounded by,
\begin{equation}
\EX \left[ C_{SR}\right] \geq  (\frac{1}{4}-\epsilon)\log{\log{L}}(1-o(1))\\
\end{equation}
\textit{where $\epsilon$ is a small positive constant and $o(1) \rightarrow 0$ as $L \rightarrow \infty$.}
\end{corollary}

Theorem \ref{the-Expected achievable rate of scheduling algorithm lower bound} and corollary \ref{cor-expected system sum-rate is lower bound using scheduling} indicate that indeed, as the number of users grows, the system's sum-rate grows as well, making scheduling not only mandatory but worthwhile.

In order to have a comparison with the best scheduling policy one can attain, we present an upper bound on the expected performance of any scheduling algorithm and its scaling law, at the limit of large $L$. 

\begin{theorem}\label{the-Expected sum-rate of scheduling algorithm upper bound}
\textit{The expected system's sum-rate of any scheduling algorithm designed for a CF system with a single relay, is upper bounded by the following,}
\begin{equation*}
	\EX\left[C_{SR}\right] \leq \frac{1}{2}\log{\left(1+\text{P}\left(2\ln{L}-\ln{\ln{L}}-2\ln{\Gamma\left(\frac{1}{2}\right)}+\frac{\gamma}{2}+o(1)\right)\right)},
\end{equation*}
\textit{where $\gamma$ is the Euler-Mascheroni constant.}
\textit{Thus, the expected sum-rate of any scheduling algorithm scales at most as $\frac{1}{2}\log{\log{L}}$.}
\end{theorem}

Theorem \ref{the-Expected sum-rate of scheduling algorithm upper bound} and corollary \ref{cor-expected system sum-rate is lower bound using scheduling} show that Algorithm \ref{algo-scheduling algorithm for all transmission} is asymptotically optimal, and we have
\begin{corollary}\label{cor-optimality of algorithm 1}
\textit{Algorithm \ref{algo-scheduling algorithm for all transmission} attains the optimal scaling law of the expected sum-rate, which is $O(\log\log{L})$.}
\end{corollary}

The above results suggest that the scaling law of the CF sum-rate is similar to the scaling law which multiuser diversity schemes achieve \cite{yoo2006optimality},\cite{kampeas2014capacity},\cite{shmuel2018performance}. In a way, this can be expected since in schemes exploiting multiuser diversity, the scheduler seeks a user (or users) which has the highest channel gain in each transmission slot, where, as we will see in the sequel, our scheduling paradigm also searches for a group of users which have high channel gains. However, these users are not necessarily the ones with the highest gain among all users. Moreover, we point out that, in terms of fairness and delay considerations, our suggested scheduling algorithm for CF promises, with very high probability, the shortest completion time for decoding all users' messages by letting more than one user to use the channel simultaneously. This is with contrast with basic multiuser diversity schemes which may results with longer periods of completion time. We elaborate on this is Section \ref{sec-Completion Time and the Value of k}.

Finally, when extending the above results to multiple relays, a major challenge is finding a subset of transmitters that will allow \emph{all relays} to successfully decode non-trivial linear combinations. In Section \ref{sec-Multiple Relays}, we prove that indeed good schedules exist, satisfying all relays simultaneously, and give heuristic schedules to allow for fast completion time with multiple relays.


\section{CF With a Large Number of Users }\label{Sec-CF_with_large_users}

In this section, we analyze the system's behavior, under the CF scheme, for the regime of a large number of users and a single relay. This analysis will be the foundation for the theorems given in Section \ref{sec-Main_results} and the motivation for the scheduling given in section \ref{Sec-Scheduling_in_CF} 
We assume that the relay seeks a coefficient vector which maximizes the achievable rate (minimize $f$ as described in \eqref{equ-Optimal a vector in quadratic form}) in a given slot. Though this may not be the optimal choice that maximizes \eqref{equ-Problem statment}, it gives an upper bound on the achievable rate of the relay in that slot.

\subsection{Minimization of the quadratic form $f$}
Examining the matrix $\b{G}$ in the minimization problem \eqref{equ-Optimal a vector in quadratic form}, one can notice that as $L$, the number of transmitters, grows, the positive diagonal elements grow very fast relative to the off-diagonal elements. Specifically, each diagonal element is a $\chi^2_L$ r.v. minus a $\chi^2_1$ r.v., whereas the off-diagonal elements are only a multiplication of two Gaussian r.vs. As $L$ grows, the former has much higher expected value compared to the later and in the limit of large $L$, the matrix $\b{G}$ will tend to a diagonal matrix if properly normalized. 
Examples of a specific $\b{G}$, for a certain realization of the channel vector $\b{h}$, are presented in Figure \ref{fig-Gmatrix}, for different dimensions. 

Consider now the quadric form \eqref{equ-Optimal a vector in quadratic form} we wish to minimize. Roughly speaking, any choice of $\b{a}$ that is not a unit vector will add more than one (probably) positive element from the diagonal of $\b{G}$ to it together with possibly negative off-diagonal elements. However, when $L$ is large, the off-diagonal elements have lesser effect compared to the diagonal ones. Therefore, intuitively, one would prefer to choose $\b{a}$ to be a unit vector.  In the remainder of this section, we make this argument formal.

The minimization function $f(\b{a})=\textbf{a}^T\textbf{G}\textbf{a}$ can be written as 

\begin{equation*}
\begin{aligned}
\b{a^TGa}&=\sum_{i=1}^{L}(1-\text{P}(\|\b{h}\|^2-h_i^2))a_i^2   -  2\sum_{i=1}^{L}\sum_{j=1}^{i-1}\text{P}h_ih_ja_ia_j \\
	       &=\|\b{a}\|^2 + \text{P}\sum_{i=1}^{L}\sum_{j=1}^{i-1} (h_ia_j-h_ja_i)^2\\
	       &=\|\b{a}\|^2 + \text{P}(\|\b{a}\|^2\|\b{h}\|^2-(\b{a}^T\b{h})^2).
\end{aligned}
\end{equation*}


\begin{figure}[t]
    \centering
    \begin{subfigure}[b]{0.23\textwidth}
        \includegraphics[width=0.95\textwidth]{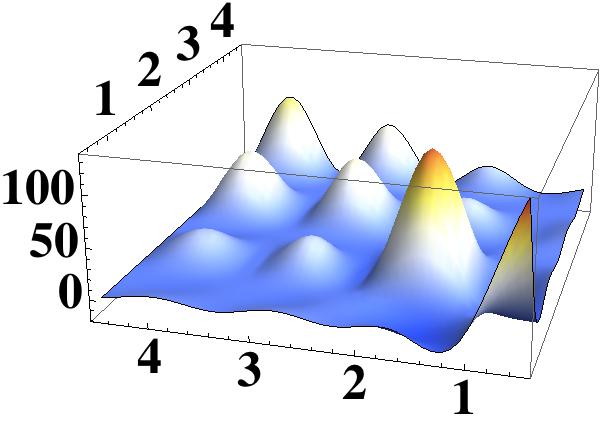}
        \caption{$L=4$}
    \end{subfigure}
    ~
    \centering
    \begin{subfigure}[b]{0.23\textwidth}
        \includegraphics[width=0.95\textwidth]{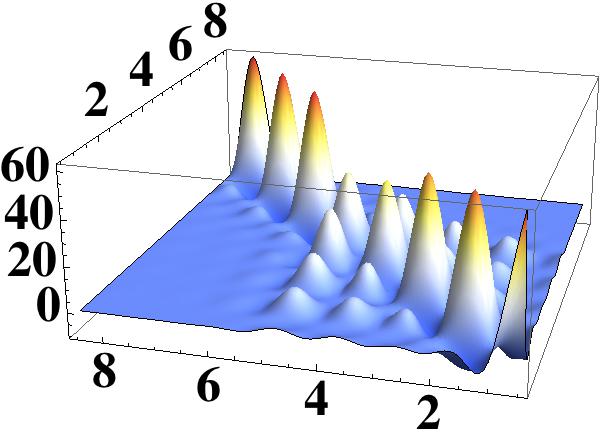}
        \caption{$L=8$}
    \end{subfigure}
    ~
    \begin{subfigure}[b]{0.23\textwidth}
        \includegraphics[width=0.95\textwidth]{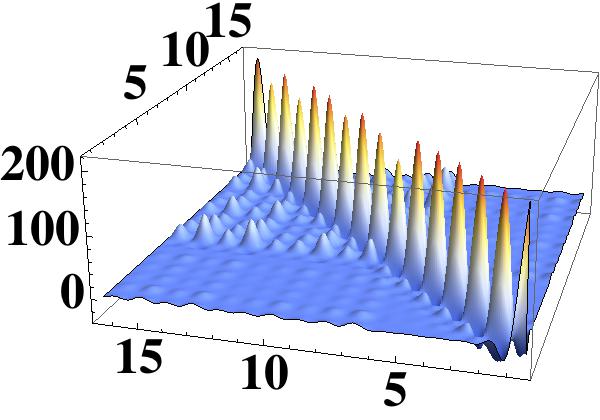}
        \caption{$L=16$}
    \end{subfigure}
    ~
    \begin{subfigure}[b]{0.23\textwidth}
        \includegraphics[width=0.95\textwidth]{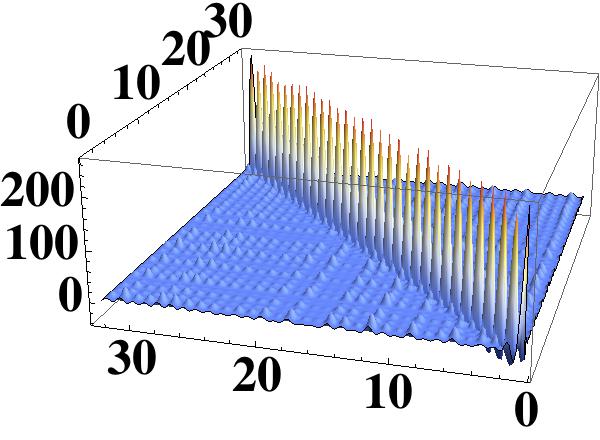}
        \caption{$L=32$}
    \end{subfigure}
    \caption{An example for the magnitude of the elements of $G$ for different dimensions (i.e., different values of $L$), for $\text{P}=10$. The graphs depict a single realization for each $L$, and were interpolated for ease of visualization.}
    \label{fig-Gmatrix}
\end{figure}

We wish to understand when will a relay prefer a unit vector over any other non-trivial vector $\b{a}$. Specifically, since $\b{a}$ is a function of the random channel $\b{h}$, we will compute the probability that a certain non-trivial $\b{a}$ will minimize $f$ compared to a unit vector. In particular, the best unit vector among all other unit vectors. 
We thus wish to find the probability
\begin{equation}\label{equ-Probability of choosing a unit vector}
	P_r( f(\b{a}) \leq \min_i{f(\b{e}_i)})=P_r\left( \|\b{a}\|^2 + \text{P}(\|\b{a}\|^2\|\b{h}\|^2-(\b{a}^T\b{h})^2\right) 
\leq \min_i\{1+\text{P}\left(\|\b{h}\|^2-h_i^2)\}\right),
\end{equation}
where $\b{e}_i$ is any unit vector of size $L$ with $1$ at the $i$-th entry and zero elsewhere, and $\b{a}$ is \emph{any} integer valued vector that is \emph{not} a unit vector. 
Note that the right and left-hand sides of the inequality in \eqref{equ-Probability of choosing a unit vector} are dependent, hence direct computation of this probability is not trivial. Still, this probability can be evaluated exactly noting that the angle between $\b{a}$ and $\b{h}$ is what mainly affects it. Formally, we give the following Lemma.

\begin{lemma}\label{lem-distribution of the squared cosine of the phase}
The distribution of $\frac{(\b{a}^T\b{h})^2}{\|\b{a}\|^2\|\b{h}\|^2}$, which is the squared cosine of the angle between an integer vector $\b{a}$ and a standard normal vector $\b{h}$, both of dimension L, is $Beta(\frac{1}{2},\frac{L-1}{2})$.
\end{lemma}
\begin{proof}
Let $\b{Q}$ be an orthogonal rotation matrix such that $\b{Q}\b{a}=\b{a'}$, where $\b{a'}$ is co-linear with the basis vector $e_1$. That is, $\b{a'}=(\|\b{a}\|,0,...,0)$. Define $\b{h'}=\b{Q}\b{h}$. Note that $\b{h'}$ is a standard normal vector since $E[\b{h'}]=E[\b{Q}\b{h}]=0$, and $\b{QIQ}^T=\b{QQ}^T=\b{I}$. We have
\begin{equation}
\begin{aligned}
	\frac{(\b{a}^T\b{h})^2}{\|\b{a}\|^2\|\b{h}\|^2}&=\frac{(\b{a}^T\b{h})^2}{(\b{a}^T\b{a})(\b{h}^T\b{h})}=\frac{(\b{a}^T\b{Q}^T\b{Qh})^2}{(\b{a}^T\b{Q}^T\b{Qa})(\b{h}^T\b{Q}^T\b{Qh})}\\
	&=\frac{((\b{Qa})^T\b{Qh})^2}{((\b{Qa})^T\b{Qa})((\b{Qh})^T\b{Qh})}=\frac{((\b{Qa})^T\b{Qh})^2}{\|\b{Qa}\|^2\|\b{Qh}\|^2}\\
	&=\frac{\|\b{a'}\|^2(\b{e_1}^T\b{h'})^2}{\|\b{a'}\|^2\|\b{h'}\|^2}=\frac{(\b{e_1}^T\b{h'})^2}{\|\b{h'}\|^2}=\frac{{h'}_1^2}{\|\b{h'}\|^2}.
\end{aligned}
\end{equation}
Therefore,
\begin{equation}
\cos^2{\theta}=\frac{{h'}_1^2}{{h'}_1^2+{h'}_2^2+...+{h'}_L^2}.
\end{equation}
This expression can be represented as $\frac{W}{W+V}$, where $W={h'}_1^2$ is a $\chi^2_1$ r.v. and $V=\sum_{i=2}^L{h'}_i^2$ is a $\chi^2_{L-1}$ r.v. independent of $W$. This ratio has a $Beta(a,b)$ distribution, with $a=\frac{1}{2}$ and $b=\frac{L-1}{2}$ \cite{walck1996hand}. Note that $a$ and $b$ correspond to the degrees of freedom of $W$ and $V$.
\end{proof}

We can now give the proof for Theorem \ref{the-Probability for having a unit vector as the maximaizer}.
\begin{proof} [Proof of Theorem \ref{the-Probability for having a unit vector as the maximaizer}]
According to equation \eqref{equ-Probability of choosing a unit vector}, we have,
\small
\begin{align*}
	P_r( f(\b{a}) \leq \min_i{f(\b{e}_i)})&=P_r\left(\|\b{a}\|^2 +\text{P}\left(\|\b{a}\|^2\|\b{h}\|^2-(\b{a}^T\b{h})^2\right) \leq \min_i\{1+\text{P}(\|\b{h}\|^2-h_i^2)\} \right)\\
	&=P_r\left(\frac{1-\|\b{a}\|^2}{\text{P}} + \|\b{h}\|^2 - \max_i{h_i^2} - \|\b{a}\|^2\|\b{h}\|^2 + (\b{a}^T\b{h})^2 \geq 0 \right)\\
	&\overset{(a)}{\leq} P_r\left(\|\b{h}\|^2 - \|\b{a}\|^2\|\b{h}\|^2 + (\b{a}^T\b{h})^2 \geq 0 \right)\\
	&=P_r\left(\frac{1}{\|\b{a}\|^2} - 1 + \frac{(\b{a}^T\b{h})^2}{\|\b{a}\|^2\|\b{h}\|^2} \geq 0 \right)\\
	&=P_r\left(\frac{(\b{a}^T\b{h})^2}{\|\b{a}\|^2\|\b{h}\|^2} \geq 1- \frac{1}{\|\b{a}\|^2} \right)\\
	&\overset{(b)}{=}1- I_{\Phi(\b{a})}\left(\frac{1}{2},\frac{L-1}{2}\right),
\end{align*}
\normalsize
where $(a)$ follows from removing the negative terms ($\|\b{a}\|^2>1$) and $(b)$ follows from Lemma \ref{lem-distribution of the squared cosine of the phase} with $\Phi(\b{a})=1- \frac{1}{\|\b{a}\|^2}$.
\end{proof}

The bound on the probability given in Theorem \ref{the-Probability for having a unit vector as the maximaizer} consists of a complicated analytic function $I_{\Phi(\b{a})}(\cdot)$. Hence, Corollary \ref{cor-Probability for having a unit vector as the maximaizer goes to one} includes a simplified bound which avoids the use of $I_{\Phi(\b{a})}(\cdot)$, yet keeps the nature of the result in Theorem \ref{the-Probability for having a unit vector as the maximaizer}. The proof of Corollary \ref{cor-Probability for having a unit vector as the maximaizer goes to one} is based on the following lemma.

\begin{lemma}\label{lem-lower bound for the distribution of the squared cosine of the phase}
The CDF of $\frac{(\b{a}^T\b{h})^2}{\|\b{a}\|^2\|\b{h}\|^2}$ can be lower bounded by the CDF of the minimum of $\left(\floor*{\frac{L}{2}}-1\right)$ i.i.d. uniform random variables in $[0,1]$.
\end{lemma}
\begin{proof}
We start by assuming that $L$ is even. The case of odd $L$ will be dealt with later. From Lemma \ref{lem-distribution of the squared cosine of the phase}, the r.v. $\frac{(\b{a}^T\b{h})^2}{\|\b{a}\|^2\|\b{h}\|^2}$ has the same distribution as $\frac{h_1^2}{\|\b{h}\|^2}$, that is, for any $0 \leq \alpha \leq 1$,
\begin{align*}
	P_r\left(\frac{(\b{a}^T\b{h})^2}{\|\b{a}\|^2\|\b{h}\|^2} \leq \alpha \right)&=P_r\left(\frac{h_1^2}{\|\b{h}\|^2} \leq \alpha \right)\\
	&\geq P_r\left(\frac{h_1^2 + h_2^2}{\|\b{h}\|^2} \leq \alpha \right)\\
	&= P_r\left(\frac{h_1^2 + h_2^2}{(h_1^2 + h_2^2)+...+(h_{L-1}^2 + h_L^2)} \leq \alpha \right)\\
	&= 1-(1-\alpha)^{\frac{L}{2}-1},
\end{align*}
where the last line is due to the observation that $\frac{h_1^2 + h_2^2}{\|\b{h}\|^2}$ can be represented as $\frac{W}{W+V}$, where $W=h_1^2 + h_2^2$ and $V=\sum_{i=3}^{L} h_i^2$ are independent exponential r.vs. Note that $V$ is essentially a sum of $\frac{L}{2}-1$ independent pairs. This ratio is distributed as the minimum of $\left(\frac{L}{2}-1\right)$ i.i.d. uniform $[0,1]$ random variables \cite[Lemma 3.2]{jagannathan2006efficient},\cite{kampeas2018ergodic}. This is since the ratio can be interpreted as the proportion of the waiting time from the first arrival to the $\frac{L}{2}$ arrival of a Poisson process. 

In case $L$ is an odd number, we can increase the term in the proof by replacing it with $\frac{h_1^2 + h_2^2}{\|\b{h}\|^2 - h_L^2}$, resulting with a ratio which is distributed as the minimum of $(\frac{L-1}{2}-1)$ i.i.d. uniform random variables in the same manner.
\end{proof} 

Using Lemma \ref{lem-lower bound for the distribution of the squared cosine of the phase} we now give the proof for Corollarly \ref{cor-Probability for having a unit vector as the maximaizer goes to one}.
\begin{proof} [Proof of Corollary \ref{cor-Probability for having a unit vector as the maximaizer goes to one}]
\small
\begin{equation}\label{equ-upper bound on the probability for not having a unit vector}
\begin{aligned}
	 P_r(f(\b{a}) \leq \min_i{f(\b{e}_i)})&\overset{(a)}{\leq} P_r\left(\frac{(\b{a}^T\b{h})^2}{\|\b{a}\|^2\|\b{h}\|^2} \geq 1- \frac{1}{\|\b{a}\|^2} \right)\\
	&\overset{(b)}{\leq} \left(1-\left(1- \frac{1}{\|\b{a}\|^2} \right) \right)^{\floor*{\frac{L}{2}}-1}\\
	&\leq   \left(\frac{1}{\|\b{a}\|^2} \right)^{\frac{L-1}{2}-1}\\
	&=   e^{-LE_1(L)},
\end{aligned}
\end{equation}
\normalsize

where $(a)$ and $(b)$ follow from Lemmas \ref{lem-distribution of the squared cosine of the phase} and \ref{lem-lower bound for the distribution of the squared cosine of the phase}, respectively, $E_1(L)=(1-\frac{3}{L})\log{\|\b{a}\|}$ and $\|\b{a}\|^2>1$. 
\end{proof}

As mentioned, Corollary \ref{cor-Probability for having a unit vector as the maximaizer goes to one} refers to the probability that a non-trivial \emph{fixed} $\b{a}$ will be the maximizer of $f$. Thus, to consider all possible coefficient vector and to analyze $P_r(\overline{\b{e}})$, the probability which any other vector that is not a unit vector was chosen, at the limit of large $L$, one needs to consider the set of all $\b{a}$ vectors in the search domain. 

In \cite{sahraei2014compute}, a polynomial-time algorithm for finding the optimal coefficient vector $\b{a}$ was given. The complexity result derives from the fact that the \textit{cardinality of the set of all $\b{a}$ vectors} which are possible is upper bounded by $2L\left(\ceil{\sqrt{1+\text{P}\|\b{h}\|^2}}+1\right)$. That is, any vector which is not in this set has zero probability to be the one that maximizes the rate. Note that all unit vectors are included in this set. We define a new set without the unit vectors by $\mathcal{A}_{\b{h}}= \{ \b{a} \in \mathbb{Z}^L : \b{a} \text{  is possible}, \ \b{a}\neq \b{e}_i \ \forall i \}$. Thus, we wish to compute
\begin{equation}
P_r(\overline{\b{e}})= P_r \left( \bigcup_{\b{a}\in\mathcal{A}_{\b{h}}} \left\{f(\b{a}) \leq f(\b{e}_i) \text{ for all }i\right\} \right).
\end{equation}
Note that the probability that another \emph{unit vector} $\b{e}_j$ will have a better metric than $\b{e}_i$ is not negligible. In fact, as $L$ grows, the unit vector with the best metric is the one corresponding to the strongest user. Note also that the cardinality of $\mathcal{A}_{\b{h}}$ grows with the dimension of $\b{h}$, i.e., with $L$, and can be easily upper bounded as follows,
\begin{equation}\label{equ-Cardinality of the search domain}
|\mathcal{A}_{\b{h}}| \leq 2L\left(\ceil*{\sqrt{1+\text{P}\|\b{h}\|^2}}+1\right)-L \\
\leq  2L(\text{P}\|\b{h}\|^2+2.5).
\end{equation}

With the above definition and the upper bound on the cardinality of the search domain, we can now give the proof for Theorem \ref{the-Probability for having a unit vector as the maximaizer over all other vectors}
\begin{proof}[Proof of Theorem \ref{the-Probability for having a unit vector as the maximaizer over all other vectors}] 
\begin{align*}
\lim_{L \rightarrow \infty}P_r(\overline{\b{e}})&= \lim_{L \rightarrow \infty}  P_r \left( \bigcup_{\b{a}\in\mathcal{A}_{\b{h}}} \left\{f(\b{a}) \leq f(\b{e}_i) \text{ for all }i\right\} \right) \\
&\leq \lim_{L \rightarrow \infty} \sum_{\mathcal{A}_{\b{h}}}  P_r\left(f(\b{a}) \leq f(\b{e}_i) \text{ for all }i\right)\\
&= \lim_{L \rightarrow \infty} \sum_{\mathcal{A}_{\b{h}}}  P_r\left(f(\b{a}) \leq \min_i{f(\b{e}_i)} \right)\\
&\overset{(a)}{\leq} \lim_{L \rightarrow \infty} \sum_{\mathcal{A}_{\b{h}}} \left(\frac{1}{\|\b{a}\|^2} \right)^{\frac{L-1}{2}-1}\\
&\overset{(b)}{\leq} \lim_{L \rightarrow \infty} |\mathcal{A}_{\b{h}}| \left(\frac{1}{2} \right)^{\frac{L-1}{2}-1}\\
&\overset{(c)}{\leq} \lim_{L \rightarrow \infty} 2L(\text{P}\|\b{h}\|^2+2.5) \left(\frac{1}{2} \right)^{\frac{L-1}{2}-1}\\
&= \lim_{L \rightarrow \infty} 2L(\text{P}\sum_{i=1}^Lh_i^2+2.5) \left(\frac{1}{2} \right)^{\frac{L-1}{2}-1}\\
&= \lim_{L \rightarrow \infty} 2L^2\text{P}\left(\frac{1}{2} \right)^{\frac{L-1}{2}-1}  \frac{1}{L}\sum_{i=1}^Lh_i^2\\
&\overset{(d)}{=} \lim_{L \rightarrow \infty} 2L^2\text{P}\left(\frac{1}{2} \right)^{\frac{L-1}{2}-1} \\
&=4\text{P} \lim_{L \rightarrow \infty}  L^2 2^{-\frac{L-1}{2}} \\
&=4\text{P}\lim_{L \rightarrow \infty}  L^2 e^{-LE_2(L)} = 0,
\end{align*}
where $(a)$ follows from Corollary \ref{cor-Probability for having a unit vector as the maximaizer goes to one} and $(b)$ is true since $\b{a}$ is an integer vector which is not the unit vector. $(c)$ is due to \eqref{equ-Cardinality of the search domain}. $(d)$ follows from the strong law of large numbers - the normalized sum converges with probability one to the expected value of $\chi^2_1$ r.v. which is one. Lastly we define $E_2(L)=\frac{1}{2}(1-\frac{1}{L})\log{2}$.
\end{proof}

\begin{figure}[t]
\centering
    \includegraphics[width=0.4\textwidth]{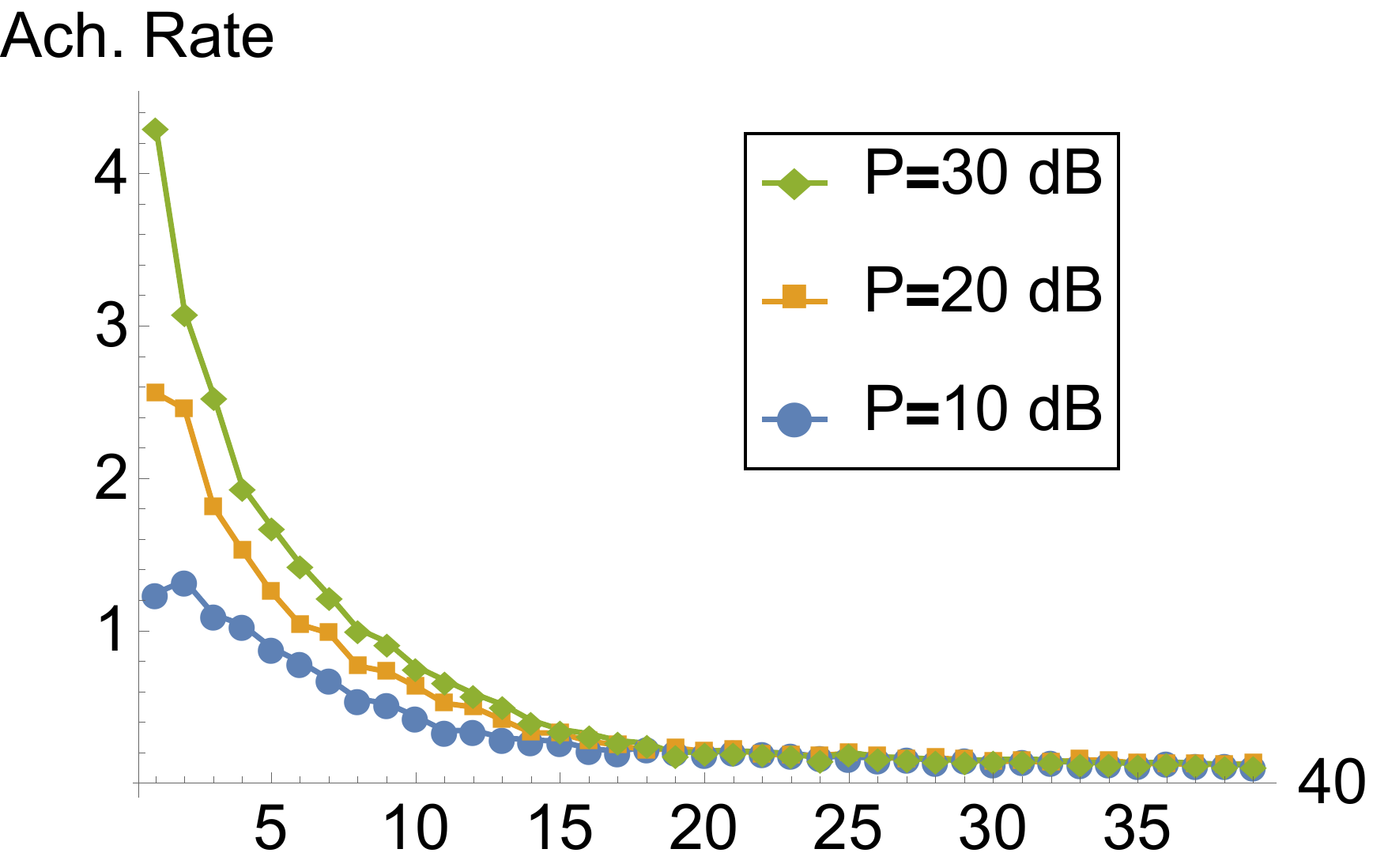}
\caption{Simulation results for the achievable rate at a relay, as given in Theorem \ref{the-Computation rate with MMSE}. The rate was plotted as a function of the number of simultaneously transmitting users $L$, for different values of $\text{P}$. In each sample, the relay chose the coefficient vector $\b{a}^{max}$.}
\label{fig-Rate_AsFunc_users}
\end{figure}

This result implies that the probability of having any non-unit vector as the rate maximizer is decreasing exponentially to zero as the number of users grows. The direct consequence is that the relay will try to decode a single message when all other transmitters are treated as interference. Theorem \ref{the-Achievable is going to zero} shows that in this case the rate of that single message has to go to zero as the number of transmitters grows. The proof is presented in Appendix \ref{AppendixA}. Moreover, Corollary \ref{the-Sum rate is going to zero} which relates to a general system consisting of $M$ relays, shows that as the number of users grows, the system's sum-rate decreases to zero as well. That is, without scheduling users, not only each individual rate is negligible; this is true for the sum-rate of the entire system as well. The proof is presented in Appendix \ref{AppendixB}.

Simulation results for the achievable rate at a relay, as a function of $L$ and for different values of $\text{P}$, are depicted in Figure \ref{fig-Rate_AsFunc_users}. Here, in each sample, the relay chose the rate maximizing coefficient vector $\b{a}^{max}$. One can observe that for large $L$, the achievable rate decreases to zero as Theorem \ref{the-Achievable is going to zero} suggests. Accordingly, following corollary \ref{the-Sum rate is going to zero}, Figure \ref{fig-SumRate no sch} depicts simulation results for the behavior of the sum-rate, for a system with $M$ relays, as a function of $L$. In fact, the curves constitute an upper bound since in the simulation we considered the achievable rates in each slot rather than the minimum rate among all slots. In addition, for decoding, we considered the $L$ linearly independent coefficient vectors (the rows of $\b{A}$) with the highest rates. As the plot shows, similar to the achievable rate at each relay, the sum-rate goes to zero as well. On the other hand, Figure \ref{fig-SumRate rand sch} presents simulation results where $k=3$ users were scheduled for transmission in each slot in a uniform manner. It is clear that even a simple scheduling policy can guarantee a non-zero rate. Note that the decrease in the sum-rate is a result of the constant number of users who were scheduled regardless of $L$. As we will see in the sequel, $k$ should also grow with $L$, or else the completion time will grow and thus the sum-rate will decreases as the figure depicts.

\begin{figure}[t]
    \centering
        \includegraphics[width=0.4\textwidth]{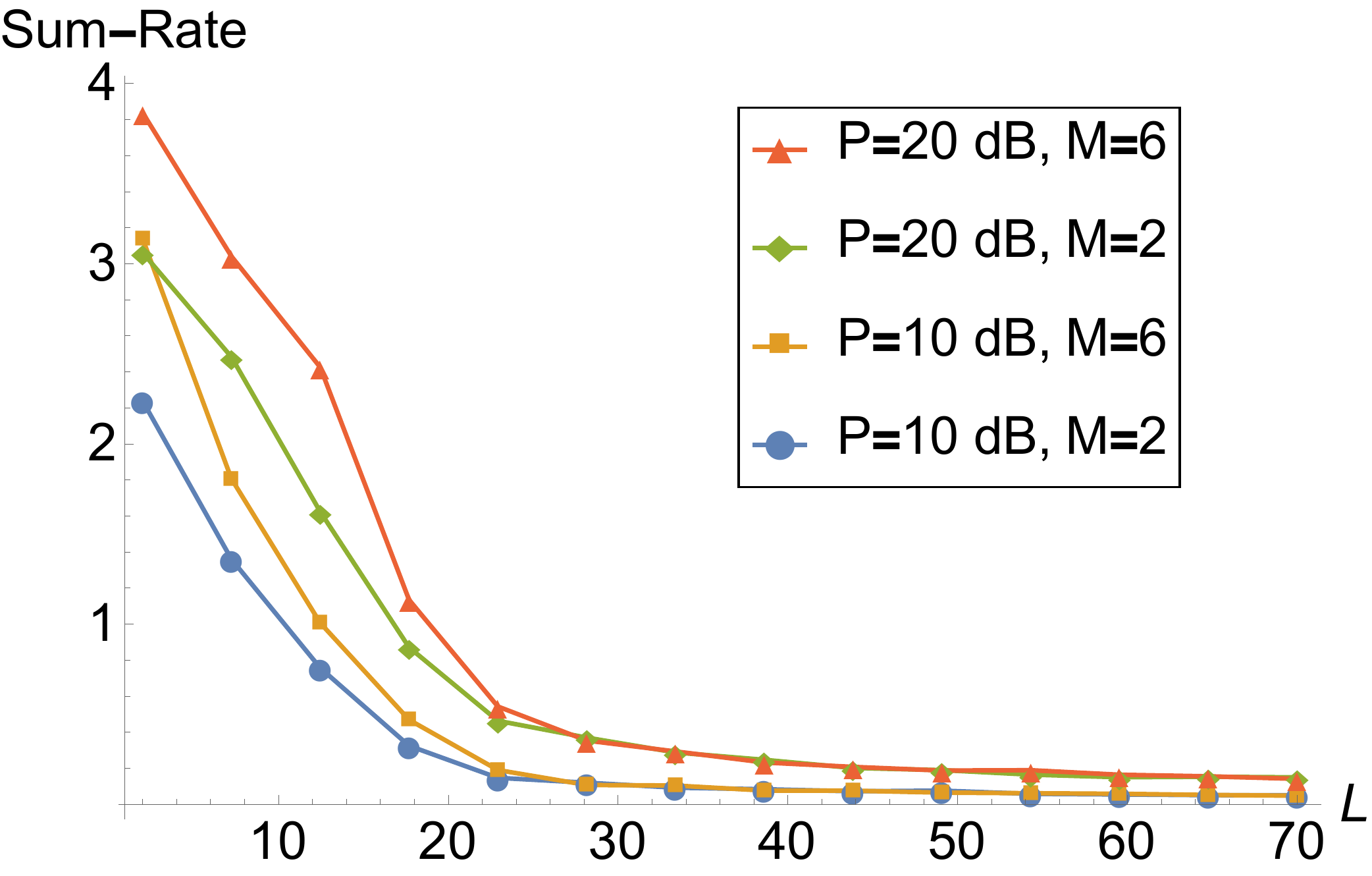}
    \caption{Simulation results for the sum-rate of a system with $M=2,6$ relays, for different values of $\text{P}$, as a function of $L$. The sum-rate decreases since each relay prefers to decode a single massage and treat all other massages as noise. Here, all the users are transmitting.}
       \label{fig-SumRate no sch}
\end{figure}

The results of this section show that the increase in the number of users can dramatically degrade the performance of a system that applies the CF coding scheme. This is in contrast to the common knowledge for the capacity of MAC which is increasing with the number of simultaneously transmitting users. The decrease in the system performance results from the approximation error of the large real channel vector by the vector of integer coefficients of the linear combination of the messages.

\begin{figure}[t]
    \centering
        \includegraphics[width=0.4\textwidth]{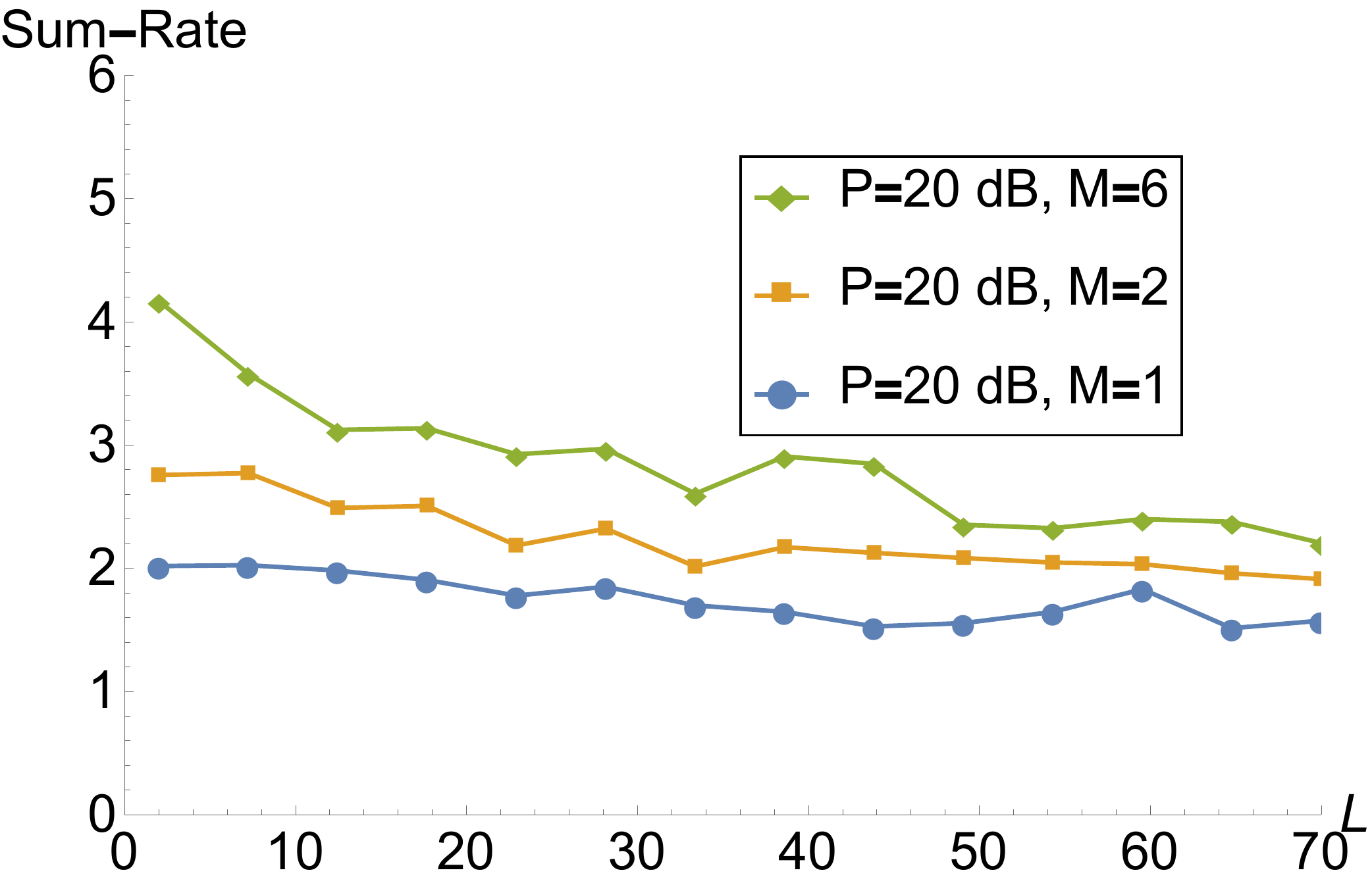}
    \caption{Simulation results for the sum-rate with a random scheduling policy. In each slot $k=3$ users were scheduled for transmission while each relay chose the rate maximizing coefficient vector $\b{a}^{max}$.}
       \label{fig-SumRate rand sch}
\end{figure}

\section{Scheduling in CF}\label{Sec-Scheduling_in_CF}


The scheduling problem presented in \eqref{equ-Problem statment with scheduling} requires maximizing the transmission rates and minimizing the number of slots, $N$, simultaneously. This should be done via jointly selecting users and the appropriate coefficient vectors. In the following subsections, we show that asymptotically with $L$, a computationally efficient solution for this problem exists. Specifically, we will show that a policy that maximizes the achievable rate in each slot will also minimize $N$ to its (asymptotic) minimum value as a direct outcome, as long as the number of scheduled users, $k$, in each slot, is $O(\log{L})$. Essentially, this will mean that the scheduling problem can be reduced to finding the subset of $k=O(\log{L})$ users $\cS_k^*(i)$ which yields the highest $\cR(\b{h}(\cS_k^*(i)),\b{a}(\cS_k^*(i)))$ in each slot $i$ \emph{separately}. That is, 
\begin{equation}\label{equ-maximization of sum rate per slot}
\begin{aligned}
\cS_k^*=\argmax_{\cS_k}\left\{\max_{\b{a}(\cS_k) \in\mathbb{Z}^k \backslash \{\textbf{0}\}}\cR(\b{h}(\cS_k),\b{a}(\cS_k))\right\}.
\end{aligned}
\end{equation}



\begin{figure*}[t]
    \centering
    \begin{subfigure}[b]{0.45\textwidth}
        \includegraphics[width=\textwidth]{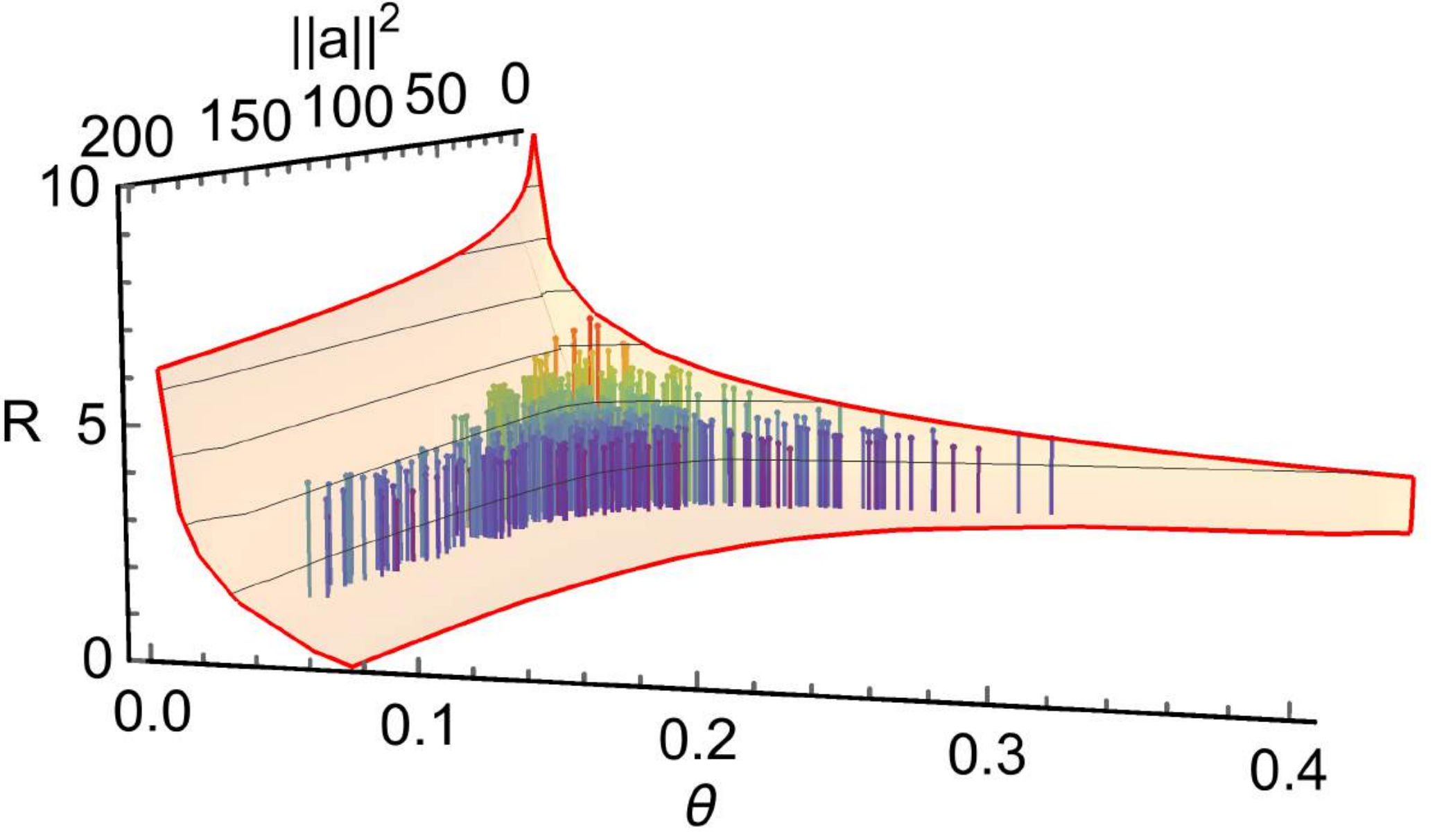}
        \caption{$L=15$}
        \label{fig-Rate_Large_P_func_theta_norm_combined_L=15}
    \end{subfigure}
    \begin{subfigure}[b]{0.45\textwidth}
        \includegraphics[width=\textwidth]{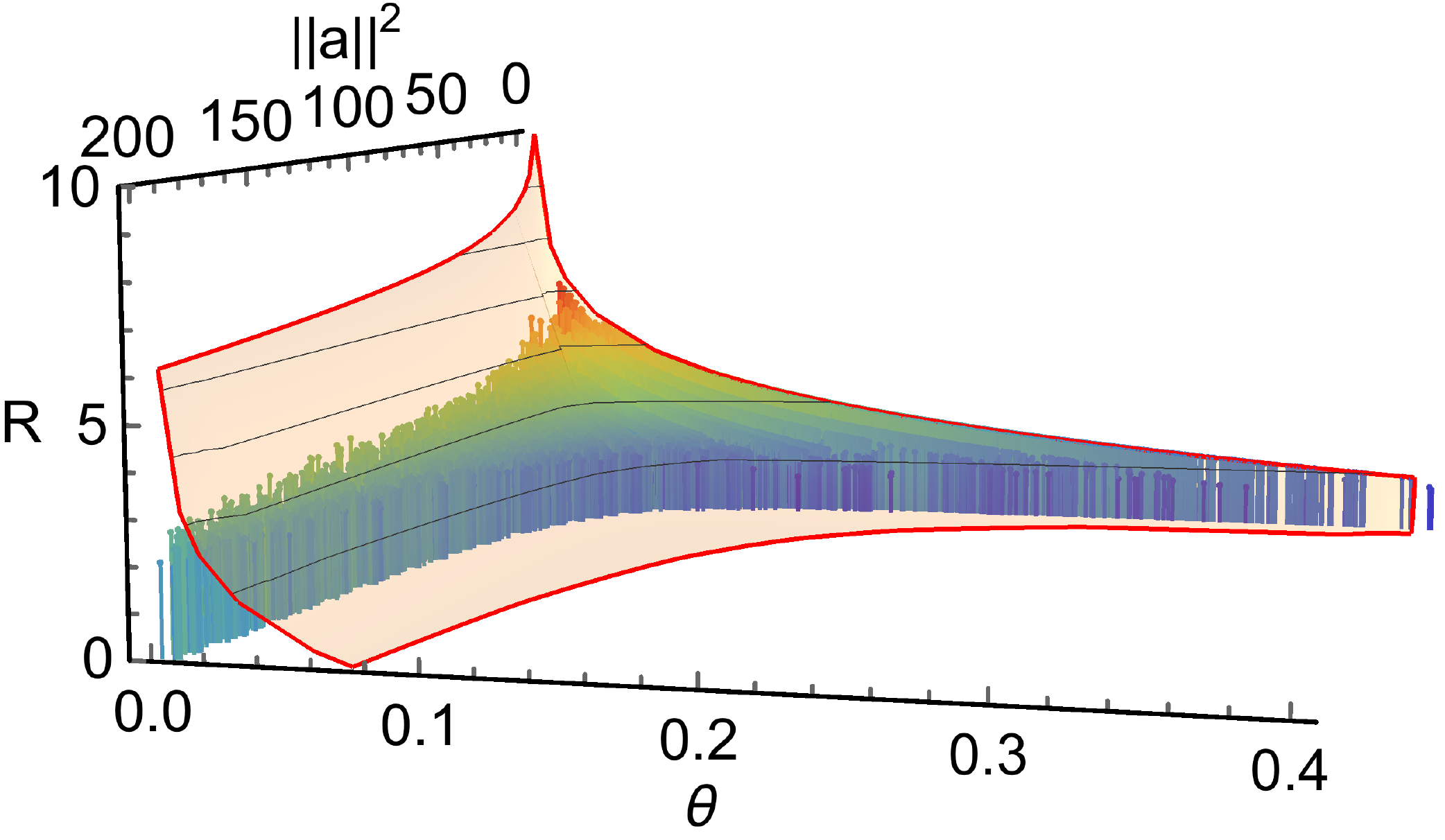}
        \caption{$L=45$}
        \label{fig-Rate_Large_P_func_theta_norm_combined_L=45}
    \end{subfigure}
    \caption{The achievable rate as a function of $\theta$, i.e., the angle between $\b{a}$ and $\b{h}$, and the squared norm of $\b{a}$. The discrete points are simulation results for the rate of each subset of users of size $k=3$ for a specific realization of the channel vector $\b{h}_L$ with $(a)L=15$, $(b)L=45$, and $\text{P}=1000$. The transparent curved plane describes an upper bound on the rate as in \eqref{equ-Achievable rate high SNR}, where, for ease of visualization, a continues function was plotted.}
    \label{fig-Rate_Large_P_func_theta_norm_combined}
\end{figure*}

Thus, we start the analysis by examining the achievable rate as a function of $\b{h}$ and $\b{a}$ and discuss the properties of the optimal schedule. Based on this analysis, we then provide a scheduling algorithm for the complete transmission of all messages, with its complexity and asymptotic guarantees. Note that since the slots are assumed to be memoryless, we omit the time index in $\cR(\b{h},\b{a})$. Furthermore, occasionally, we also omit the users' subset index when it is obvious from the context and write $\b{h}$ and $\b{a}$ as a general writing for the sake of notation simplicity. 

\subsection{Achievable Rate Under Scheduling}

The main challenge in low complexity scheduling for CF is identifying the characteristics of the channel values of a certain subset of users such that this subset would be classified as a good option for CF. Remembering the rate expression in \eqref{equ-Computation rate with MMSE}, the challenge is to identify a subset of the channel values that can be well approximated by an integer vector.

We start by exploring an upper bound on the achievable rate. This bound motivates the suggested algorithm, which will be given in the sequel. We have
\begin{align}\label{equ-Achievable rate high SNR}
\cR(\b{h},\b{a})&= \frac{1}{2} \log^+ \left( \|\b{a}\|^2- \frac{\text{P}(\b{h}^T\b{a})^2}{1+\text{P}\|\b{h}\|^2} \right)^{-1} \nonumber \\
		      &\leq \frac{1}{2} \log^+ \left( \|\b{a}\|^2- \frac{(\b{h}^T\b{a})^2}{\|\b{h}\|^2} \right)^{-1} \nonumber \\	
      		      &=\frac{1}{2} \log^+ \left( \|\b{a}\|^2- \|\b{a}\|^2\cos^2(\theta) \right)^{-1} \nonumber \\			
      		      &=\frac{1}{2} \log^+ \left( \|\b{a}\|^2\sin^2(\theta) \right)^{-1},			
\end{align}
where $\theta$ is the angle between $\b{h}$ and the chosen coefficient vector $\b{a}$. 


The behavior of the achievable rate as a function of $\theta$ and $\sqn{a}$ is depicted in Figure \ref{fig-Rate_Large_P_func_theta_norm_combined}. The discrete lines represent simulation results for $\cR(\b{h}(\cS_k),\b{a}(\cS_k))$ for each subset of size $k=3$, out of a realization of the channel vector $\b{h}_L$. That is, for each subset $\cS_k$, the optimal $\b{a}(\cS_k)$ was chosen according to \eqref{equ-Optimal a vector definition}. $\text{P}=1000$. The continuous curve is a smoothed representation of \eqref{equ-Achievable rate high SNR}. The smoothing is since, for one, $\b{a}$ is an integer vector, hence, its squared norm takes only integer values. Second, for a certain $\|\b{a}\|^2$, there are only finitely many possible choices of $\b{a}$ and thus a finite number of angles with $\b{h}$. For example, for $\|\b{a}\|^2=5$ and dimension 2 the possible vectors are only $(1,2),\ (-1,2),\ (1,-2)$ and $(-1,-2)$. That is, in this case, there are 4 possible angles with a given $\b{h}$. 

One can infer from the upper bound \eqref{equ-Achievable rate high SNR} and Figure \ref{fig-Rate_Large_P_func_theta_norm_combined} that 
the subsets of users that attain the highest rates are those with low values of $\sqn{a}$ and $\theta$. Moreover, we note that the slope of the rate as a function of $\theta$ is sharper than the slope as a function of $\sqn{a}$, with an exception for the smallest values of $\sqn{a}$. Accordingly, for any given $\b{h}_L$, we expect to use only a \emph{fixed set of coefficient vectors} with low norm values as the set the relay will choose from. Then, try to find a subset of transmitters $\cS_k$ with the smallest angle between $\b{h}(\cS_k)$ and one vector out of this set. 

Considering the above, define the set of coefficient vectors: 
\begin{equation}\label{equ-definition of the set of all one vectors}
\b{a^{\{1,k\}}}=\left\{ \b{a} \ \Big| \  \b{a}\in\Z^k, \ |a_i|=1,\ i=1,...,k \right\}.
\end{equation}
Note that $\forall \b{a} \in \b{a^{\{1,k\}}}, \ \sqn{a}=k$, and that the cardinality of $\b{a^{\{1,k\}}}$ is $2^k$. In what follows, we suggest that the scheduled subset of users would be a subset \emph{for which the relay will be able to choose the coefficient vector out of $\b{a^{\{1,k\}}}$} as its rate maximizer.
Thus, Equation \eqref{equ-maximization of sum rate per slot} can be written as, 
\begin{equation}\label{equ-maximization of sum rate per slot reduced}
\cS_k^{\b{a^{\{1,k\}}}}=\argmax_{\cS_k}\left\{\max_{\b{a}(\cS_k) \in \b{a^{\{1,k\}}} }\left\{\cR(\b{h}(\cS_k),\b{a}(\cS_k))\right\}\right\}.
\end{equation}

Note that since $\b{a}\in \b{a^{\{1,k\}}}$ refers to the coefficient vector of \emph{scheduled} users, it has no zero entries.

The following Lemma shows an important property of the optimal coefficient vector $\b{a}$ which maximizes $\cR(\b{h},\b{a})$.
\begin{lemma}\label{lem-the signs of the coefficient vector are ruled by the channel}
\textit{The optimal vector $\b{a}$ satisfies either, $sign(h_i)=sign(a_i)$ for all $i$ or  $sign(h_i) \neq sign(a_i)$ for all $i$.}
\end{lemma}
\begin{IEEEproof}
Considering the rate expression \eqref{equ-Computation rate with MMSE}, since $\sqn{a}$ does not depend on the signs, the optimal signs must maximize the inner product $(\b{h}^T\b{a})^2$. Obviously, all signs must match (or be oposite) in order to have only positive (or negative) elements in the summation of the inner product. 
\end{IEEEproof}

The usefulness of Lemma \ref{lem-the signs of the coefficient vector are ruled by the channel} is in making the inner maximization in \eqref{equ-maximization of sum rate per slot reduced} trivial, since given a subset of users $\cS_k$ with channel coefficients $\b{h}(\cS_k)$, the optimal $\b{a}(\cS_k) \in \b{a^{\{1,k\}}}$ is clear - just set the signs according to those of $\b{h}(\cS_k)$. Consequently, the following procedure is optimal for solving \eqref{equ-maximization of sum rate per slot reduced}: disregard the signs in $\b{h}_L$; find the optimal subset $\cS_k$ by considering only absolute values and finding the subset which best fits $\b{a}=(1,1,...,1)\triangleq \b{1}$; then simply set the signs of $\b{a}$ from all positive to the original signs of $\b{h}(\cS_k)$. This reduces the double optimization in \eqref{equ-maximization of sum rate per slot reduced}, with $2^k$ options in the inner one, to a much simpler optimization: 
\begin{equation}\label{equ-maximizing the rate for an all one vector}
\cS_k^{\b{a^{\{1,k\}}}}=\argmax_{\cS_k}\left\{\cR(|\b{h}(\cS_k)|,\b{1})\right\},
\end{equation}
where $|\b{h}(\cS_k)|=(|h_1|,|h_2|,...,|h_k|)$.
Thus, when searching for the optimal schedule in each slot, we significantly relax the optimization: we fix a reasonably good $\b{a}$ and search for the best $\b{h}$. As it turns out, this will be asymptotically optimal.

\subsection{Best channel for a fixed $\b{a}$}\label{subsec-Best channel for a fixed a}
 Towards the solution of \eqref{equ-maximizing the rate for an all one vector}, note that
\begin{equation}\label{equ-Scheduling optimal rate fixed a}
\begin{aligned}
\argmax_{\cS_k}\left\{\cR(\b{h}(\cS_k),\b{a})\right\} &= \argmax_{\cS_k}\left\{\frac{1}{2} \log^+ \left( \|\b{a}\|^2- \frac{\text{P}(\b{h}(\cS_k)^T\b{a})^2}{1+\text{P}\|\b{h}(\cS_k)\|^2} \right)^{-1}\right\}\\
				&=\frac{1}{2} \log^+ \left( \|\b{a}\|^2-  \argmax_{\cS_k}\left\{\frac{\text{P}(\b{h}(\cS_k)^T\b{a})^2}{1+\text{P}\|\b{h}(\cS_k)\|^2}\right\} \right)^{-1}\\
				&=\argmax_{\cS_k}\left\{\frac{\text{P}(\b{h}(\cS_k)^T\b{a})^2}{1+\text{P}\|\b{h}(\cS_k)\|^2}\right\}\\
				&=\argmax_{\cS_k}\left\{\frac{\text{P}\|\b{h}(\cS_k)\|^2\|\b{a}\|^2\cos^2(\theta)}{1+\text{P}\|\b{h}(\cS_k)\|^2}\right\}\\
				&=\argmax_{\cS_k}\left\{\frac{\cos^2(\theta)}{1+\frac{1}{\text{P}\|\b{h}(\cS_k)\|^2}}\right\}.				
\end{aligned}
\end{equation}
Thus, generally speaking, an $\b{h}$ which maximizes the achievable rate should have a high norm and a small angle with $\b{a}$. Clearly, the highest norm vector may not be the one with the smallest angle to $\b{a}$. Thus, the scheduler should seek the optimal tradeoff point to maximize the achievable rate. Our suggested scheduling algorithm, given in the next subsection, searches for this optimal tradeoff point in polynomial-time. The search relies on the following lemma which shows that, for the case of all-ones coefficient vector, sorting the channel vector $\b{h}_L$ by the elements' absolute value simplifies the search significantly. Thus, we define $\b{h}_L^s$ to be an ascending ordered vector according to $|\b{h}_L|$.
\begin{lemma}\label{lem-optimal schedule for all one vector}
The optimal subset $\cS_k$ for the all-ones vector $\b{1}$ is a subset for which $\b{h}(\cS_k)$ is $k$ consecutive elements in $\b{h}_L^s$. That is,
\begin{equation*}
\max_{\cS_k}\left\{\cR(|\b{h}(\cS_k)|,\b{1})\right\}=\max_{i}\left\{\cR(|\b{h}_i'|,\b{1})\right\},
\end{equation*}
where $\b{h}_i'= (h_{L,i}^s,...,h_{L,i+k-1}^s)$ for $i \in [1,...,L-k+1]$.
\end{lemma}
\begin{IEEEproof}
Consider the expression for the achievable rate in Theorem \ref{the-Computation rate}.
We have,
\begin{equation}
\begin{aligned}
\max_{\cS_k}\left\{\cR(|\b{h}(\cS_k)|,\b{1})\right\}&=\max_{\cS_k}\left\{ \frac{1}{2} \log^+ \left( \frac{\text{P}}{\min \limits_{\alpha \in \R} \left\{\alpha^2+\text{P}\|\alpha |\b{h}(\cS_k)|-\b{1}\|^2\right\}} \right)\right\}\\
&= \frac{1}{2} \log^+ \left( \frac{\text{P}}{\min \limits_{\cS_k} \left\{\min \limits_{\alpha \in \R} \left\{\alpha^2+\text{P}\|\alpha |\b{h}(\cS_k)|-\b{1}\|^2\right\}\right\}} \right)\\
&= \frac{1}{2} \log^+ \left( \frac{\text{P}}{\min \limits_{\alpha >0} \left\{\min \limits_{\cS_k} \left\{\alpha^2+\text{P}\|\alpha |\b{h}(\cS_k)|-\b{1}\|^2\right\}\right\}} \right),\\
\end{aligned}
\end{equation}
where in the last line we can reduce the minimization to $\alpha >0$ since for $\alpha <0$ we would increase the term for all $|\b{h}(\cS_k)|$. Therefore we need to show that for any $\alpha >0$ 
\begin{equation*}
\b{h}\left(\argmin \limits_{\cS_k} \left\{\|\alpha |\b{h}(\cS_k)|-\b{1}\|^2\right\}\right)=\b{h}_i',
\end{equation*}
for  some $i \in [1,...,L-k+1]$.

Define the sequence $\Delta_j=(\alpha h_{L,j}^s-1)^2$, for $j=1,...,L$. This sequence can be monotonic increasing, monotonic decreasing or monotonic decreasing and then monotonic increasing with $j$; it depends on the value of $\alpha h_{L,1}^s$ and $\alpha h_{L,L}^s$ with respect to $1$. For example, if $\alpha h_{L,1}^s>1$ then the sequence is monotonic increasing with $j$.
Let us choose some $\cS_k$ with $\b{h}(\cS_k)$ such that its corresponding elements in $\b{h}_L^s$ are not consecutive. Hence, w.l.o.g. assume that two elements in $\b{h}(\cS_k)$ corresponds to two elements $h_{L,i}^s$ and $h_{L,j}^s$ such that $i+1 \neq j $. Accordingly, either the choices $\b{h}_j'$ or $\b{h}_{j-k}'$ will minimize  $\left\{\|\alpha |\b{h}(\cS_k)|-\b{1}\|^2\right\}$ since in at least one of the choices we would decreased with the sequence $\Delta_i$. Note also that this is true for the choices $\b{h}_i'$ or $\b{h}_{i-k}'$
\end{IEEEproof}

\subsection{Scheduling Algorithm}

In this section, we present a polynomial-time scheduling algorithm presented as Algorithm \ref{algo-scheduling algorithm for all transmission} which is an asymptotically optimal solution for the scheduling problem as defined in \eqref{equ-Problem statment with scheduling}. The algorithm relies on the properties suggested in the previous subsections, and therefore in each slot searches the subset of users which maximizes \eqref{equ-maximizing the rate for an all one vector}. Specifically, this search is done by the subroutine Algorithm \ref{algo-scheduling algorithm per slot} which relies on Lemmas \ref{lem-the signs of the coefficient vector are ruled by the channel} and \ref{lem-optimal schedule for all one vector}. The output of Algorithm \ref{algo-scheduling algorithm for all transmission} is the set of subsets of users which should be scheduled in each slot along with the decoding matrix $\b{A}$. 

The complexity of Algorithm \ref{algo-scheduling algorithm per slot} is $O(L\log{L}+(L-k)k)$ due to the sorting of $\b{h}_L$ and the scan of $L-k$ scheduling options for which it computes the achievable rate on vectors with length $k$. Accordingly, the complexity of Algorithm \ref{algo-scheduling algorithm for all transmission} is $O\left(\left(L-1\right)\left(L\log{L}+(L-k)k\right)\right)$ which uses Algorithm \ref{algo-scheduling algorithm per slot} $L-1$ slots. With $k=O(\log{L})$, this results in $O(L^2\log{L})$. Since $O(L^2\log{L})$ is the required complexity to invert the coefficient matrix \cite{wiedemann1986solving}, the complexity of the scheduling algorithm is within this range and does not add any significant computations above the necessary order.



Simulation results of the system's expected sum-rate for Algorithm \ref{algo-scheduling algorithm for all transmission} is depicted in Figure \ref{fig-Optimal schedule CF comparisons}. The simulation results are compared with an upper bound on the expected system sum-rate which is calculated as the mean of the highest achievable rate in $L$ slots. That is, we ignore the rank restriction and take the highest rate in each transmission slot. However, since this calculation becomes prohibitively complex, the exhaustive search was done only up to $L=40$. For larger $L$ the plot merely an interpolation. The figure also includes the upper bound as given in Theorem \ref{the-Expected sum-rate of scheduling algorithm upper bound} and $\log{\log{L}}$ to reflect the scaling law anticipated by Corollary \ref{cor-optimality of algorithm 1}. One can observe that, as $L$ grows, the expected sum-rate that Algorithm \ref{algo-scheduling algorithm for all transmission} provides coincides with the curve of the optimal schedule.

\begin{remark}[Real-time algorithm]
Algorithm \ref{algo-scheduling algorithm for all transmission} attains its optimality without requiring global CSI of all slots in advance. That is, the scheduler only needs the CSI in the beginning of each slot in order to determine which $k$ users to schedule.
\end{remark}

Algorithm \ref{algo-scheduling algorithm for all transmission} completes, yet successful decoding can actually occur if both the rank of the decoding matrix $\b{A}_L$ is $L$ and the transmission rate of all users were below the minimal achievable rate at the relay among all slots. Theorem \ref{the-completion time} and Lemma \ref{lem-outage probability goes to zero} in the next subsection discuss these two critical conditions.

\begin{figure}[t]
    \centering
        \includegraphics[width=0.48\textwidth]{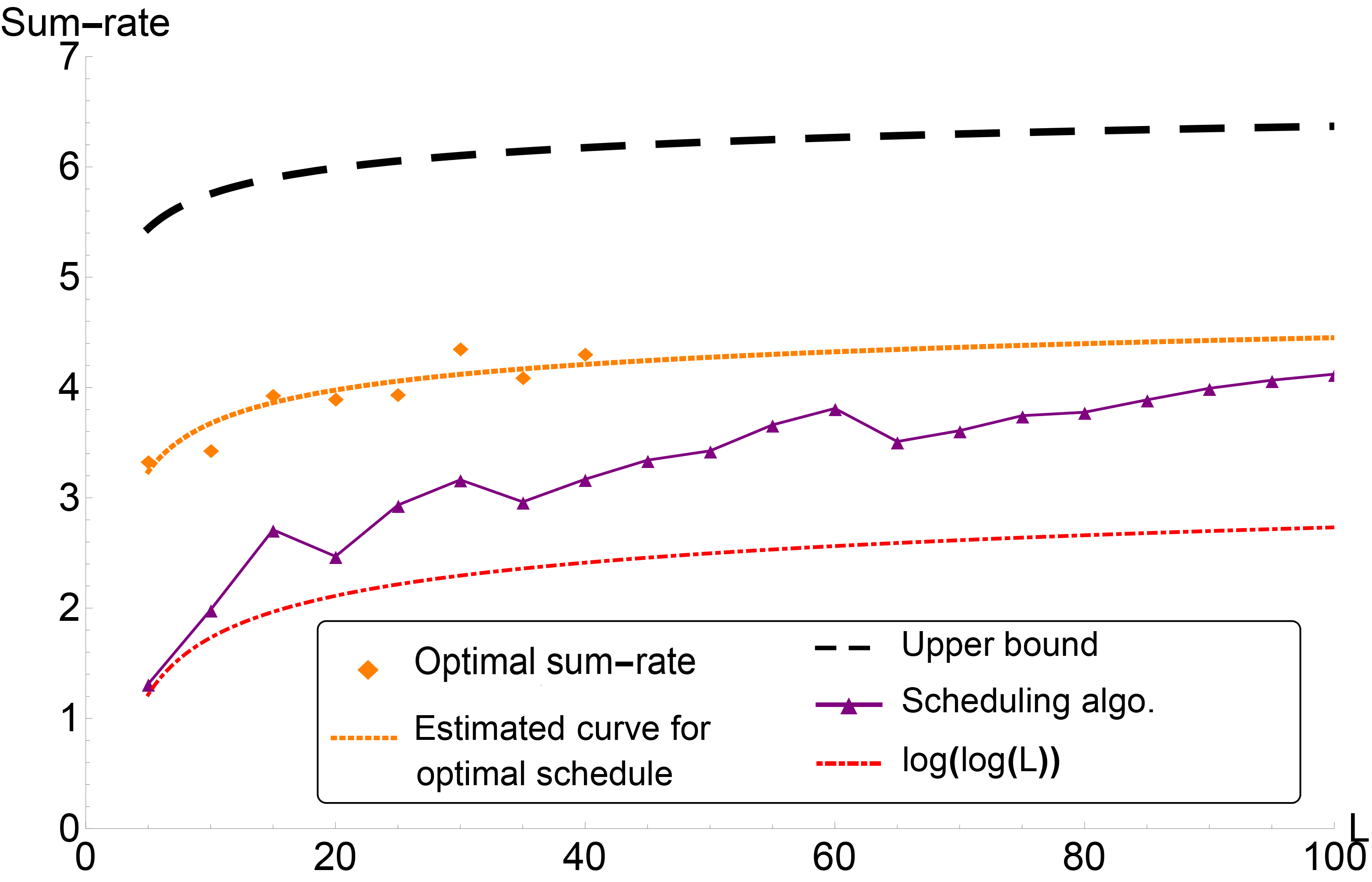}
    \caption{The system's expected sum-rate of the Algorithm  \ref{algo-scheduling algorithm for all transmission} compared with the expected sum-rate for the optimal schedule as a function of $L$ with $\text{P}=1000$. The asymptotic upper bound which is given in Theorem \ref{the-Expected sum-rate of scheduling algorithm upper bound} along with the function $\log{\log{L}}$, which is the scaling law Algorithm \ref{algo-scheduling algorithm for all transmission} achieves as given in Corollary \ref{cor-optimality of algorithm 1}, are also plotted.}      
     \label{fig-Optimal schedule CF comparisons} 
\end{figure}

\begin{algorithm}
\caption{Scheduling for CF (single relay)}
\label{algo-scheduling algorithm for all transmission}
\hspace*{\algorithmicindent} \textbf{Input:} $\b{h}_L$ \\
 \hspace*{\algorithmicindent} \textbf{Output:} $(\cS_k^N,\b{A})$ or $error$
\begin{algorithmic}[1]
\Algphase{Initialization:}
\State $\b{A} \gets \emptyset$ ; $L\gets length(\b{h}_L)$
\State $R_{min}\gets\infty; \ R^*\gets0; \ R\gets \text{The transmission rate}$
\State $k\gets \lceil\ln{L}\rceil+1$
\State $\cS_k^N \gets \emptyset$
\Algphase{Main:}
\For {$i =1; \ i\leq L-1; \ i++$}
	\State $(\cS_k^*,\b{a}^*,R^*) \gets  Algorithm 2(\b{h}_L,\b{1}^k)$
	\State Schedule users according to $\cS_k^*$ and $\cS_k^N \gets \cS_k^N \cup \cS_k^*$ 
	\State $\b{A} \gets  \b{A} \cup (\b{a}^*(\cS_k^*))_\Uparrow^L$ 
	\If {$R^* < R_{min}$}
		\State $R_{min}=\cR(\b{h}(\cS_k^*),\b{a}^*(\cS_k^*))$ 
	\EndIf
\EndFor
	\State Find $\b{e}_i$ that is linearly independent with $\b{A}$  
	\State $\b{A}  \gets \b{A} \cup \b{e}_i$ 
	\State Schedule user $i$
	\If {Rank$(\b{A})==L \ \&\& \ R<R_{min}$}
		\State Solve linear system According to $\b{A}$
		\State \Return $(\cS_k^N,\b{A})$
	\Else
		\State \Return $(error)$ 
	\EndIf
\end{algorithmic}
\end{algorithm}
 
\begin{algorithm}
\caption{Finding optimal schedule for $\b{a}=\b{1}$ per slot}
\label{algo-scheduling algorithm per slot}
\hspace*{\algorithmicindent} \textbf{Input: $(\b{h}_L,\b{a})$} \\
 \hspace*{\algorithmicindent} \textbf{Output: $(\cS_k^*,\b{a}^*,R^*)$}
\begin{algorithmic}[1]
\Algphase{Initialization:}
\State $\b{h}_L^s \gets$ sort according to $Abs(\b{h}_L)$ 
\State $\b{h}_L^I \gets \text{ordering of } \b{h}_L^s \text{ in } \b{h}_L$ \Comment{indices vector}
\State $\b{h}_L^{sign} \gets \b{h}_L^s./Abs(\b{h}_L^s)$ \Comment{element-wise devision}
\State $k \gets length(\b{a})$
\State $R^* \gets  0$; 
\State $i^* \gets -1$
\State $\cS_k^* \gets \emptyset$
\State $\b{a}^* \gets \emptyset$
\Algphase{Main:}
\For {$i =1; \ i\leq L-k; \ i++$}
	\State $\b{h} \gets  \b{h}_L^s(i:i+k)$
	\If {$\cR(|\b{h}|,\b{a}) > R^*$}
		\State $i^* \gets i$
		\State $R^* \gets \cR(|\b{h}|,\b{a})$
	\EndIf
\EndFor
\If {$i^* = -1$} 
	\Return $(\{1:k\},\b{a})$ \Comment{zero rate}
\Else
	\State $\cS_k^* \gets \{\b{h}_L^I(i^*):\b{h}_L^I(i^*)+k-1\}$
	\State $\b{a}^* \gets \b{a}.*\b{h}_L^{sign}(i^*:i^*+k-1)$ 
\EndIf\\
\Return $(\cS_k^*,\b{a}^*,R^*)$
\end{algorithmic}
\end{algorithm}

\subsection{Successful Decoding of Algorithm \ref{algo-scheduling algorithm for all transmission}}\label{sec-Completion Time and the Value of k}

We first link the number of transmitters, the number of scheduled users, and the coefficient vectors to the rank of $\b{A}_L$. Recall that in each slot a coefficient vector $\b{a}(\cS_k)\in \b{a^{\{1,k\}}}$ is chosen by the relay for the scheduled subset of users $\cS_k$. This vector is then mapped to an $L$-length vector using the function $(\cdot)_\Uparrow^L$. Thus, is each slot, a vector from the set $\Ss_{L,k}^{\{-1,1\}}$ is added to $\b{A}_L$. Since the users' channel coefficients are $i.i.d.$, the scheduled subset $\cS_k$ in each slot is uniformly distributed among all subsets. In addition, so do the signs of the channel's coefficients. Accordingly, the distribution of the vectors from $\Ss_{L,k}^{\{-1,1\}}$ which are added to $\b{A}_L$ is also uniform.

Theorem \ref{the-completion time} below shows that with $O(\log{L})$ scheduled users in each slot, we can indeed achieve a full rank with high probability after $L$ slots.

\begin{theorem}\label{the-completion time}
Assume $k=O(\log{L})$ and the coefficient vectors for the scheduled subset $\cS_k$ are drawn uniformly from $\b{a^{\{1,k\}}}$ for $L-1$ slots. Then, there exists a unit vector such that selecting it at the $L$th slot results in $rank(\b{A}_L)=L$ with probability $1-o(1)$.
\end{theorem}
In other words, Theorem \ref{the-completion time} asserts that by uniformly selecting vectors from $\b{a^{\{1,k\}}}$ for $L-1$ slots, then completing the matrix with a single unit vector, one has a full rank with high probability. To prove Theorem \ref{the-completion time}, we first give the lemma bellow.

\begin{lemma}\label{lem-rank of i in i slots}
Let $\b{A}_n$ be the decoding matrix at slot $n$, where each row $i\leq n$ is a vector from $\Ss_{L,k}^{\{-1,1\}}$ corresponding to $(\b{a}(\cS_k(i)))_\Uparrow^L$, where $\b{a}(\cS_k(i))$ was drawn uniformly from $\b{a^{\{1,k\}}}$. If 
\begin{equation}
n<\left(1-\frac{e^{-k}}{\ln{2}}-o\left(e^{-k}\right) \right)L
\end{equation}
then $rank(\b{A}_n)=n$ with probability $1-o(1)$.
\end{lemma}
\begin{IEEEproof}
The proof follows from \cite[Theorem 3]{calkin1997dependent}, also appearing in \cite{mazumdar2014update} with a similar formulation. Specifically, the results in \cite{calkin1997dependent,mazumdar2014update} consider random vectors over the binary field $\F_2^L$ with exactly $k$ ones, i.e., $\Ss_{L,k}^{\{1\}}$.

Define the matrix $\b{B}_n=(\b{A}_n \mod 2)$, where the modulo operation is element-wise. The modulo operation acts as a mapping between vectors in $\Ss_{L,k}^{\{-1,1\}}$, which have non-zero elements at certain positions, to vectors in $\Ss_{L,k}^{\{1\}}$ which have only ones at the same positions. Every $-1$ becomes $1$. Since the rows of $\b{A}_n$ are uniformly distributed from $\Ss_{L,k}^{\{-1,1\}}$ and for every vector in $\Ss_{L,k}^{\{1\}}$ there are exactly $2^k$ corresponding vectors in $\Ss_{L,k}^{\{-1,1\}}$, the probability remains uniform on the elements of $\Ss_{L,k}^{\{1\}}$. Thus, the rows of $\b{B}_n$ are uniformly distributed from $\Ss_{L,k}^{\{1\}}$  and according to \cite{calkin1997dependent}, if $n<\left(1-\frac{e^{-k}}{\ln{2}}-o\left(e^{-k}\right) \right)L$ then $rank(\b{B}_n)=n$ with probability $1-o(1)$.\\ 
Thus, since $rank(\b{A}_n)\geq rank(\b{B}_n)$, the result follows. Finally, since we will need a more precise expression for the $o\left(e^{-k}\right)$ term we note that this is in fact \cite{calkin1997dependent}, \\ $\frac{1}{2 \log{2}} \left(k^2-2k+\frac{2 k}{\log{2}}-1\right)e^{-2k} - O\left(k^4\right)e^{-3k}=o\left(e^{-k}\right)$.
\end{IEEEproof}

\begin{IEEEproof}[Proof of Theorem \ref{the-completion time}]
According to Lemma \ref{lem-rank of i in i slots}, a rank of $L-1$ can be achieved in $L-1$ slots with probability $1-o(1)$ as long as
\begin{equation}
L-1<\left(1-\frac{e^{-k}}{\ln{2}}-o\left(e^{-k}\right) \right)L.
\end{equation}
The above inequality reduces to
\begin{equation}
\frac{1}{L}>\frac{e^{-k}}{\ln{2}}+o\left(e^{-k}\right), 
\end{equation}
which is satisfied as long as
\begin{equation}\label{equ-the k value}
k>\ln{L}+\ln{2}.
\end{equation}
Note that the term $o\left(e^{-k}\right)$ stated in Lemma \ref{lem-rank of i in i slots} decays very fast with $k$.
Letting $k=\lceil\ln{L}\rceil+1$ meets the requirement above. 
In the $L$-th slot, the scheduler may schedule a single user (a unit vector) which will be the missing degree of freedom for achieving rank $L$. Such vector always exists since otherwise, this means that all the unit vectors are spanned by $\b{A}_{L-1}$ which is a contradiction.   
\end{IEEEproof}



Theorem \ref{the-completion time} shows that indeed the rank of $\b{A}_L$ is $L$ with probability $1-o(1)$, which is asymptotically the shortest completion time for a single relay model.

We now show that the rate restriction \eqref{equ-rate restriction for successful decoding} is also satisfied. That is, if $R$ is the transmission rate of all the users, then the following must be satisfied, 
\begin{equation}
R<\min_{n:1,...,L-1} \cR(\b{h}(n),\b{a}(n)).
\end{equation}
Accordingly, we define an outage scenario when the above condition is not satisfied. Formally,
\begin{equation}\label{equ-outage definition}
P_{out}(R) = P_r\left(\min_{n:1,...,L-1}\{ \cR(\b{h}^*(n),\b{a}^*(n))\} < R \right),
\end{equation}
where $\b{h}^*(n)$ and $\b{a}^*(n)$ are the channel and coefficient vectors of the scheduled subset $\cS_k^*(n)$ in time slot $n$.
Note that we assume that in the last slot a single user $i$ is scheduled; thus it's rate is restricted to a SISO Gaussian channel capacity, i.e., $R<\frac{1}{2}\log(1+h_i^2\text{P})$. The following lemma shows that using the lower bound on the achievable rate, as given in Theorem \ref{the-Expected achievable rate of scheduling algorithm lower bound}, allows a zero outage probability at the limit of large $L$.

\begin{lemma}\label{lem-outage probability goes to zero} 
If the transmission rate is set to be $R=\left(\frac{1}{4}-\epsilon\right)\log{\log{L}}$ where $\epsilon$ can be chosen to be arbitrary small, the outage probability vanishes. That is,
\begin{equation}
\lim_{L \rightarrow \infty} P_{out}\left( R\right) =0.
\end{equation}
\end{lemma}
The proof is given in Appendix \ref{AppendixH}.

\subsection{Proofs for the Asymptotic Guarantees}

We are now ready to present the proofs for Theorems \ref{the-Expected achievable rate of scheduling algorithm lower bound} and \ref{the-Expected sum-rate of scheduling algorithm upper bound} which rely on the previous subsections. 

\begin{IEEEproof}[Proof of Theorem \ref{the-Expected achievable rate of scheduling algorithm lower bound}]
The expected rate that can be achieved in each slot by Algorithm \ref{algo-scheduling algorithm for all transmission} is as follows. Note that the time index is omitted due to the independence between the slots.
\begin{align*}
\EX\left[\cR_{ach}^{sch}\right]&=\EX\left[\max_{\cS_k}\left\{\max_{\b{a}(\cS_k) \in\mathbb{Z}^k \backslash \{\textbf{0}\}}\left\{\cR(\b{h}(\cS_k),\b{a}(\cS_k))\right\}\right\}\right]\\
&\overset{(a)}{\geq}\EX\left[\max_{\cS_k}\left\{\max_{\b{a}(\cS_k) \in \b{a^{\{1,k\}}} }\left\{\cR(\b{h}(\cS_k),\b{a}(\cS_k))\right\}\right\}\right]\\
&\overset{(b)}{\geq}\EX\left[\max_{\cS_k}\left\{\cR(|\b{h}(\cS_k)|,\b{1})\right\}\right],
\end{align*}
$(a)$ and $(b)$ follow since we reduce the search domain as explained in \eqref{equ-maximization of sum rate per slot reduced} and \eqref{equ-maximizing the rate for an all one vector}, respectively. This enables the following steps,
\begin{align}\label{equ-expected sum-rate}
	\EX\left[\max_{\cS_k}\left\{\cR(|\b{h}(\cS_k)|,\b{1})\right\}\right] &=\EX\left[\max_{\cS_k}\left\{ \frac{1}{2} \log^+ \left( k- \frac{\text{P}\left(|\b{h}(\cS_k)|^T\b{1}\right)^2}{1+\text{P}\|\b{h}(\cS_k)\|^2} \right)^{-1} \right\}\right] \nonumber\\
	&=\EX\left[\frac{1}{2} \log^+ \left( k- \max_{\cS_k}\left\{ \frac{\text{P}\left(|\b{h}(\cS_k)|^T\b{1}\right)^2}{1+\text{P}\|\b{h}(\cS_k)\|^2} \right\} \right)^{-1}\right] \nonumber\\
	&\overset{(a)}{\geq} \EX\left[\frac{1}{2} \log^+ \left( k-  \frac{\text{P}\left(\b{h'}^T\b{1}\right)^2}{1+\text{P}\|\b{h'}\|^2}  \right)^{-1}\right] \nonumber\\
	&\overset{(b)}{\geq} \frac{1}{2} \log^+ \left( k- \EX\left[ \frac{\text{P}\left(\b{h'}^T\b{1}\right)^2}{1+\text{P}\|\b{h'}\|^2}\right]  \right)^{-1} \nonumber\\
	&= \frac{1}{2} \log^+ \left( k- \EX\left[ \frac{\text{P}k\sqn{h'}\cos^2(\theta')}{1+\text{P}\|\b{h'}\|^2} \right]  \right)^{-1}, 
\end{align}
where $(a)$ is by choosing some specific $\b{h}(\cS_k')$, denoted in short by $\b{h'}$ and $(b)$ follows from Jensen's inequality. 

As section \ref{subsec-Best channel for a fixed a} suggests, the optimal schedule should be a subset of users with a high norm channel vector and a small angle between its channel vector and the corresponding coefficient vector. Thus, define the values $u(L)$ and $\delta(L)$ such that $\b{h'}$ maintains
\begin{equation}\label{equ-the k good users}
u\leq | h_i' | \leq u+\delta, \ \forall i.
\end{equation}
With this definition, we are able to bound the parameters for a good schedule. The values of $u(L)$ and $\delta(L)$ can help tune the norm (by taking a high value of $u$) and the angle with $\b{1}$ (by taking a small value of $\delta$) to attain a high rate (see Section \ref{subsec-Best channel for a fixed a}). 

Let $\xi$ denote the event of having at least $k$ elements in $\b{h}_L(n)$ with values in the interval $[-(u+\delta),-u]\cup[u,u+\delta]$ for all $n$, i.e., \emph{in each of the $L$ slots}. 
We thus write the last equation in \eqref{equ-expected sum-rate} as follows, 

\begin{equation}\label{equ-the rate split to probabilities}
\small
\begin{aligned}
		&= \frac{1}{2} \log^+ \left( k- \left(\EX\left[ \frac{\text{P}k\sqn{h'}\cos^2(\theta')}{1+\text{P}\|\b{h'}\|^2} \ \Big| \xi \right] P_r(\xi) +\EX\left[ \frac{\text{P}k\sqn{h'}\cos^2(\theta')}{1+\text{P}\|\b{h'}\|^2} \ \Big| \bar{\xi} \right] (1-P_r(\xi)) \right) \right)^{-1}\\	
		&\geq \frac{1}{2} \log^+ \left( k- \EX\left[ \frac{\text{P}k\sqn{h'}\cos^2(\theta')}{1+\text{P}\|\b{h'}\|^2} \ \Big| \xi \right] P_r(\xi) \right)^{-1}.
\end{aligned}
\end{equation}\normalsize
Under $\xi$, we can lower bound $\sqn{h'}$ and $\cos^2(\theta')$ as follows,
\begin{equation}\label{equ-lower bound on h' and cos}
\begin{aligned}
		&\sqn{h'} \geq ku^2;\\
		&\cos^2(\theta')=\frac{\left( \sum_{i=1}^{k} h_i' \right)^2}{k\sqn{h'}} \geq \frac{k^2u^2}{k^2(u+\delta)^2}=  \frac{u^2}{(u+\delta)^2}.
\end{aligned}
\end{equation}

The probability $P_r(\xi)$ can be computed and lower bounded as follows. Consider the probability to find in a certain slot $k$ users satisfying \eqref{equ-the k good users}. This probability follows a binomial distribution with probability of success $p(u,\delta)=2(\Phi(u+\delta)-\Phi(u))$ where $\Phi$ is the CDF of the normal distribution. Accordingly, since the $L$ time slots are independent we have,

\begin{equation}\label{equ-lower bound on probability xi}
\begin{aligned}
P_r(\xi)&=\left(\sum_{i=k}^{L} {L \choose i} p(u,\delta)^i(1-p(u,\delta))^{L-i}\right)^L\\
	&=\left(1-\sum_{i=0}^{k-1} {L \choose i} p(u,\delta)^i(1-p(u,\delta))^{L-i}\right)^L\\
	&\geq \left(1-e^{-\frac{1}{2p(u,\delta)}\frac{(Lp(u,\delta)-(k-1))^2}{L}}\right)^L,
\end{aligned}
\end{equation}
where the last row follows from the Chernoff's bound for the lower tail, which requires that $Lp(u,\delta) \geq k-1$.
By setting 
\begin{equation}\label{equ-u(L) and delta(L) values}
u(L) =\sqrt{2\ln\frac{\delta\sqrt{L}}{\sqrt{2 \pi}}}-\delta \quad \text{and} \quad  \delta(L)=\frac{1}{\ln{L}}
\end{equation}
the requirement is satisfied and the probability $P_r(\xi)$ goes to one with $L$. This is since for for $L\geq4$ we have

\begin{equation}\label{equ-p(u,delta) lower bound}
\begin{aligned}
p(u,\delta)&= 2(\Phi(u+\delta)-\Phi(u)) \\
&= \frac{2}{\sqrt{2\pi}} \int_{u}^{u+\delta} e^{-\frac{t^2}{2}}dt\\
&\geq \delta  \frac{2}{\sqrt{2\pi}} e^{-\frac{(u+\delta)^2}{2}}\\
&= \delta  \frac{2}{\sqrt{2\pi}} e^{-\frac{\left(\sqrt{2\ln\frac{\delta\sqrt{L}}{\sqrt{2 \pi}}}-\delta+\delta\right)^2}{2}}\\
&= \frac{2}{\sqrt{L}},
\end{aligned}
\end{equation}
and since, $ \frac{2}{\sqrt{L}} \geq \frac{\lceil\ln{L}\rceil}{L}$ the Chernoff's bound requirement is satisfied when choosing $k$ according to \eqref{equ-the k value}.  
On the other hand, 

\begin{equation}\label{equ-limit of P_r(xi) goes to one}
\begin{aligned}
\lim_{L \rightarrow \infty} P_r(\xi)&
 \geq \lim_{L \rightarrow \infty} \left(1-e^{-\frac{1}{2p(u,\delta)}\frac{(Lp(u,\delta)-(k-1))^2}{L}}\right)^L \\
& \overset{(a)}{\geq} \lim_{L \rightarrow \infty} 1-L e^{-\frac{\sqrt{L}}{4}\frac{(2\sqrt{L}-\lceil\ln{L}\rceil)^2}{L}} \\
&=1,
\end{aligned}
\end{equation}
where $(a)$ is due to \eqref{equ-p(u,delta) lower bound} and Bernoulli's inequality which state that, $(1-x)^L\geq 1-Lx$ for $L\geq 1$ and $0 \leq x \leq 1$. 

As $L$ grows, the values of $u(L)$ and $\delta(L)$ are increasing and decreasing, respectively. Specifically, the slow increase of $u(L)$ promises that the channel gains of the scheduled users will increase, while the decrease of $\delta(L)$ improves the match to the all-ones coefficient vector.

Substituting \eqref{equ-lower bound on h' and cos} and \eqref{equ-lower bound on probability xi} in \eqref{equ-the rate split to probabilities} we have, 
\begin{equation}\label{equ-lower bound on the sum-rate final expression} 
\small
\begin{aligned}	
		&\frac{1}{2} \log^+ \left( k- \EX\left[ \frac{\text{P}k\sqn{h'}\cos^2(\theta')}{1+\text{P}\|\b{h'}\|^2} \ \Big| \xi \right] P_r(\xi) \right)^{-1}\\
		&\quad\quad\geq \frac{1}{2}  \log^+ \left( k \left(1-\frac{\text{P}ku^4}{(u+\delta)^2(1+\text{P}ku^2)}\left(1-e^{-\frac{1}{2p(u,\delta)}\frac{(Lp(u,\delta)-(k-1))^2}{L}}\right)^L \right)\right)^{-1}\\
		&\quad\quad\overset{(a)}{=} \frac{1}{2}  \log^+ \left( k \left(1-\frac{\text{P}ku^4}{(u+\delta)^2(1+\text{P}ku^2)}\left(1-o(1)\right) \right)\right)^{-1}.
\end{aligned}
\end{equation}
\normalsize

It can be verified (see Appendix \ref{AppendixC}) that the scaling law of \eqref{equ-lower bound on the sum-rate final expression} is indeed $\frac{1}{4}\log{\log{L}}$, which completes the proof. 
\end{IEEEproof}

\begin{IEEEproof}[Proof of Corollary \ref{cor-expected system sum-rate is lower bound using scheduling}]
The system expected sum-rate can be lower bounded as follows,
\begin{equation}
\begin{aligned}
\EX\left[C_{SR}\right] &=\E\left[  \frac{RL}{N} \right]  \\
&\geq \E\left[  \frac{RL}{N} \Big | N=L \right]P_r(N=L) \\
&\overset{(a)}{=}\left( \frac{R(L-1)}{L}+\frac{R^{SISO}}{L}  \right)(1-o(1)) \\
&\geq \frac{R(L-1)}{L} (1-o(1)) \\
&=(\frac{1}{4}-\epsilon)\log{\log{L}} (1-o(1)).
\end{aligned}
\end{equation}
$(a)$ follows due to Theorem \ref{the-completion time} and the fact that Algorithm \ref{algo-scheduling algorithm for all transmission} schedules in the last slot a single user with rate of the SISO channel. The other slots are with rate $R$. Finally, Lemma \ref{lem-outage probability goes to zero} guarantees that if one sets $R=(\frac{1}{4}-\epsilon)\log{\log{L}}$ for any small positive $\epsilon$, an innovative linear combination with rate $R$ would be successfully decoded in each of the first $L-1$ slots with probability that goes to 1 as $L$ grows.
  \end{IEEEproof}

\begin{IEEEproof}[Proof of Theorem \ref{the-Expected sum-rate of scheduling algorithm upper bound}]
In \cite{nazer2016diophantine}, the following universal upper bound on the achievable rate was given
\begin{equation}\label{equ-universal upper bound on the achievable rate}
 	\cR(\b{h},\b{a}^{opt})\leq \frac{1}{2}\log{(1+\text{P}\max_i\{h_i^2\})},
\end{equation}
where $\b{h}$ is any channel vector of dimension $k$ and $\b{a}^{opt}$ is the coefficient vector which maximizes the achievable rate. Considering \eqref{equ- system sum-rate definition}, we have

\begin{align*}
\EX\left[C_{SR}\right] &= \EX \left[ \frac{RL}{N}\right]\\
&\overset{(a)}{\leq} \EX\left[ \min_{n=1,...,N} \cR(\b{h}(n),\b{a}(n)) \right]\\
&\leq \EX\left[ \cR(\b{h},\b{a}^{opt}) \right]\\
&\overset{(b)}{\leq} \EX\left[\frac{1}{2}\log{(1+\text{P}\max_{i:1..k}\{h_{i}^2\})}\right]\\
&\overset{(c)}{\leq} \frac{1}{2}\log{(1+\text{P}\EX\left[\max_{i:1..L}\{h_{Li}^2\}\right])}\\
&\overset{(d)}{=} \frac{1}{2}\log{\left(1+\text{P}\left(2\ln{L}-\ln{\ln{L}}-2\ln{\Gamma\left(\frac{1}{2}\right)} +\frac{\gamma}{2}+o(1)\right)\right)}.
\end{align*}
Where in $(a)$ we set $N=L$ which is the minimal number of transmission slots and in $(b)$ we used the universal upper bound in \eqref{equ-universal upper bound on the achievable rate}. $(c)$ follows due to Jensen's inequality and the consideration of all values of $\b{h}_L$. $(d)$ follows from the asymptotic results for the expectation of the maximum value in a $\chi^2$ $i.i.d.$ random vector of dimension $L$ in the limit of large $L$ \cite[Table 3.4.4]{embrechts2013modelling}. 
It can be verified (see Appendix \ref{AppendixD}) that the scaling law is $O(\frac{1}{2}\log{\log{L}})$, which completes the proof.
\end{IEEEproof}

Corollary \ref{cor-expected system sum-rate is lower bound using scheduling} and Theorem \ref{the-Expected sum-rate of scheduling algorithm upper bound} show that the upper and lower bounds on the expected sum-rate scale as $O(\log{\log{L}})$, which results in Corollary \ref{cor-optimality of algorithm 1}.

\subsection{Distributed Scheduling}\label{sec-Distributed Scheduling}
The asymptotic guarantees presented in the previous subsection rely on the existence of $k$ users in a predefined interval of channel gain values. In fact, Algorithm \ref{algo-scheduling algorithm per slot} only presents an efficient way for the scheduler to find these channel gains and the corresponding users in each slot. Considering a single relay model, the obvious choice for the scheduler is the relay itself, since the relay has the channel vector from all users to it; in addition, the relay performs the decoding of the received messages. Nevertheless, one can devise a distributed threshold-based algorithm where only users with channel gain values that exist in a predefined interval can transmit in a certain slot. The interval upper and lower values should be computed in a way that promises, with very high probability, that there exist exactly $k$ users in it and that these users are a good choice for the all one coefficient vector. To analyze such a distributed scheme and compute the specific interval, one can emend the analysis from the proof of Theorem \ref{the-Expected achievable rate of scheduling algorithm lower bound} to fit these requirements. Alternately, one can use other statistical tools, such as extreme value theory and point process analysis, similar to other works that considered distributed threshold-based scheduling algorithms (e.g., \cite{kampeas2014capacity, shmuel2018performance}). We note that the last schedule where the scheduler picks the user which completes the full rank of the decoding matrix can be replaced by a few more random schedules as done in the first $L-1$ steps.

\section{Multiple Relays}\label{sec-Multiple Relays}

In the general model of $M\geq1$ relays and $L$ users, each relay $m$ sees a different channel vector $\b{h}_m$ between itself and all the users. Hence, a certain schedule of users which is good for a certain relay may not be the right choice for the other relays. However, we show that, as $L$ grows, there exist $k$ users that are simultaneously good for all relays.

Recall the scheduling problem given in \eqref{equ-Problem statment with scheduling}, where we wish to maximize $C_{SR}$. In each slot, $M$ coefficient vectors are added to the decoding matrix simultaneously. These vectors depend on the subset of users $\cS_k$ which was scheduled, making the scheduling problem very hard to solve and analyze. It requires taking into consideration all the $M$ channel vectors a specific schedule compels and find the appropriate coefficient vectors while ensuring that the maximal DoF from these vectors is obtained. In addition, the rates must also satisfy \eqref{equ-rate restriction for successful decoding}, so that the decoding matrix can be solved. Nevertheless, we show that if one employs a scheduling scheme which follows Algorithm \ref{algo-scheduling algorithm for all transmission} guidelines, good schedules can be found in each slot. Moreover, assuming the number of relays $M=O(1)$, and using heuristics solutions for a fast completion time, simulations show that indeed a pre-log gain of $M$ to the expected system sum-rate can be attained. We conclude the discussion by showing why this is possible for small $M$, yet fails when $M$ is too large.

The suggested scheme searches for $k$ users, all having channel gains satisfying 
\begin{equation}\label{equ-the k good users M relays}
u\leq | h_{m} | \leq u+\delta, \quad 1\leq m \leq M.
\end{equation}
We then use coefficient vectors from $\b{a^{\{1,k\}}}$. We do not present the actual scheme to find these simultaneously good $k$ users, however, Lemma \ref{lem-probability of finding $k$ users multiple relays} bellow asserts that such users can indeed be found\footnote{One can extend the polynomial-time algorithm presented in this work for a single relay, which scans the channel vector a single relay sees in ascending order, to a parallel scan of all channel vectors simultaneously, while searching the simultaneously good $k$ users. We omit the details.}.

We start with the following corollary, which states that one can always attain the performance of a single relay system.
\begin{corollary} \label{cor- lower bound on sum-rate for a general system}
The asymptotic expected system sum-rate for a general system consisting $M\geq1$ relays is lower bounded by the following,
\begin{equation}
\EX\left[C_{SR}^{M}\right] \geq \EX\left[C_{SR}\right].
\end{equation}
\end{corollary}
The above follows immediately if one employs Algorithm \ref{algo-scheduling algorithm for all transmission} on the general system while the scheduling decision is taken according to a single leading relay. Any contribution by the other relays can only reduce the completion time by adding innovative coefficient vectors to $\b{A}$. Note that the choice of the leading relay can be optimized.
\begin{remark}[Scheduling according to a single relay]
Since the scheduled users are expected to transmit with a relatively high rate, which is tailored to the channel vector the leading relay sees, relays $2,...,M$ are not expected to contribute much. A possible adaptation that will enable the other relays to contribute is to employ superposition coding. In superposition coding, one can divide a message into several messages, each belonging to a different level. By doing so, the leading relay, which can decode equations with a high rate, will be able to decode all parts. while the other relays, which have a lower rate, would be able to decode only messages belonging to certain levels. In CF this may be done by superimposing lattice codes that are scaled according to the power constraints. Further explanation and results can be found in \cite{nazer2011compute,nazer2009structured}.   
\end{remark}

To show that, in fact, good schedules can be found for all relays simultaneously, return to the analysis of the scheme suggested in Section \ref{Sec-Scheduling_in_CF}. Since the channels between the relays and the users are $i.i.d.$, the probability of finding at least $k$ users which are good \emph{simultaneously for all relays} is $P_r(\xi)^M$. This probability goes to 1 as $L$ grows with the same choice of $u(L)$ and $\delta(L)$ given in Theorem \ref{the-Expected achievable rate of scheduling algorithm lower bound}. Specifically, we have the following. 

\begin{lemma}\label{lem-probability of finding $k$ users multiple relays}
The probability of finding $k$ users in each transmission slot, such that their channel gains \emph{for all relays} are in $[-(u+\delta),-u]\cup[u,u+\delta]$, tends to $1$ as $L$ grows. That is,
\begin{equation}
\lim_{L\rightarrow \infty} P_r(\xi)^M=1.
\end{equation}
\end{lemma}
\begin{IEEEproof}
A simple extension of \eqref{equ-lower bound on probability xi} results in
\begin{equation}
\begin{aligned}
	\lim_{L\rightarrow \infty} P_r(\xi)^M
	=\lim_{L\rightarrow \infty} \left(\sum_{i=k}^{L} {L \choose i} p(u,\delta)^i(1-p(u,\delta))^{L-i}\right)^{LM}
	=1.
\end{aligned}
\end{equation}
We emphasize that the probability tends to one even for $M=O(L)$.
\end{IEEEproof}

Lemma \ref{lem-probability of finding $k$ users multiple relays} promises that asymptotically with $L$, a favourable group of users exists, and therefore the expected achievable rate, in each slot and at each relay, can be lower bounded by the result in Theorem \ref{the-Expected achievable rate of scheduling algorithm lower bound}. This implies that the rate of each linear combination scales as $O(\log{\log{L}})$ which is the optimal scaling law (Theorem \ref{the-Expected sum-rate of scheduling algorithm upper bound}).

\begin{figure}[t]
    \centering
        \includegraphics[width=0.44\textwidth]{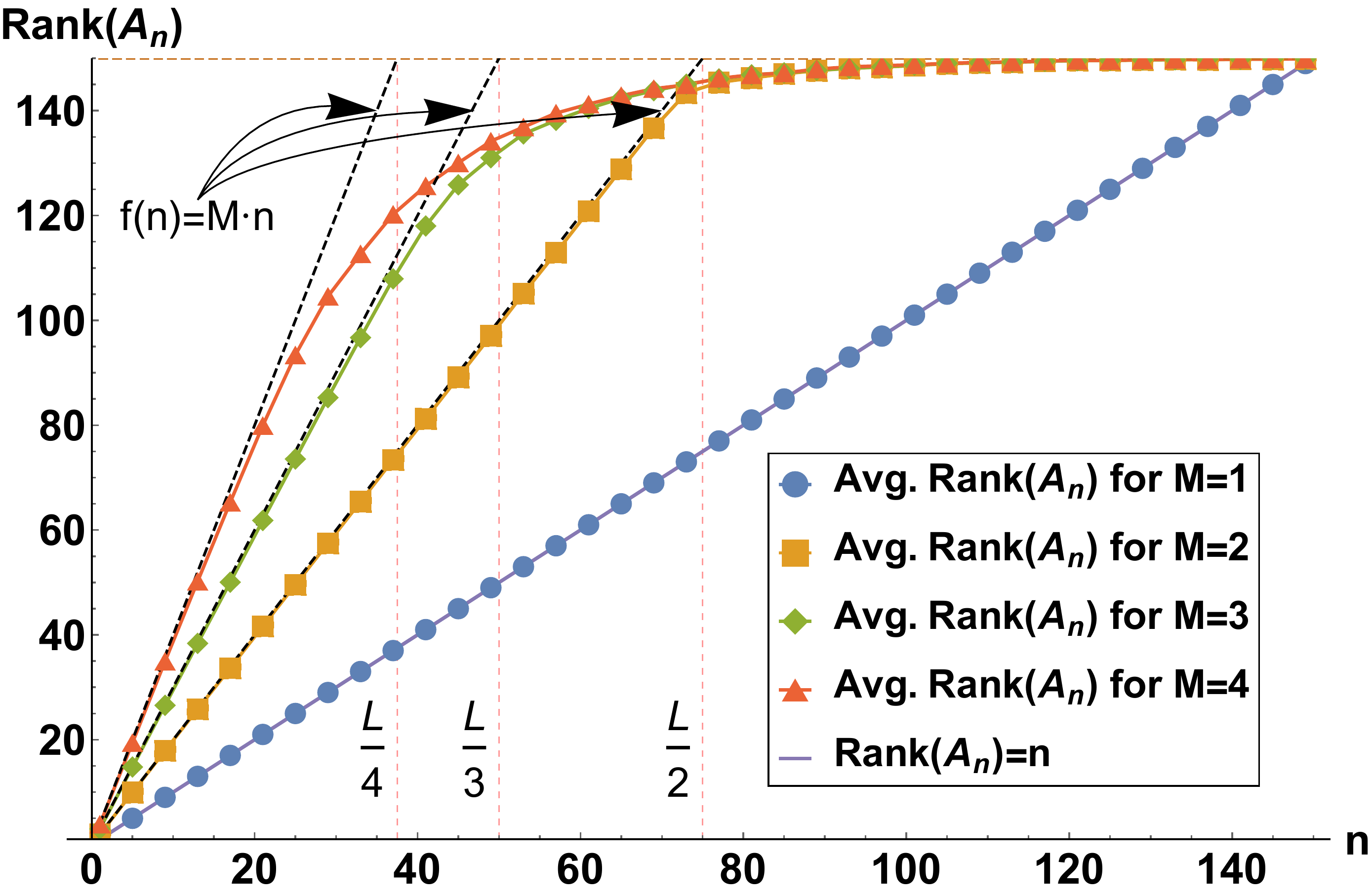}
    \caption{The average rank of $\b{A}_n$ as a function of the transmission slots $n$ for different numbers of relays $M=\{1,2,3,4\}$ and for $L=150$ transmitters.}
       \label{fig-rank_A_as_func_n_M_relays}
\end{figure}

On the other hand, the crux of the suggested scheme lies in the following. Since there are $M$ relays, in each slot $n$, $M$ coefficient vectors are added to $\b{A}_n$. Denote by $\b{A}'_n$ the sub-matrix added to  $\b{A}_n$ at slot $n$. The rows of $\b{A}'_n$ are $(\b{a}_m(\cS_k(n)))_\Uparrow^L$ for $m=1,...,M$, respectively. That is, $\b{A}_n=(\b{A}'_1,\b{A}'_2,...,\b{A}'_n)$. 

On average, the rank of each such $\b{A}'_n$ is at least $\min\{M,k\}-2$. This follows from a result given in \cite[Theorem 4]{lucani2009random} which states that the expected number of random vectors that are needed to be drawn uniformly from $\F_2^k$ to attain $k$ DoF is at most $k+2$. Since the rows of $\b{A}'_n$ were chosen uniformly from $\b{a^{\{1,k\}}}$ before the expansion $(\cdot)_\Uparrow^L$, this result applies. However, the average rank of $\b{A}_n$ does not grow linearly at rate $M$ with $n$ since, in each slot $n$, the relays choose the coefficient vectors with respect to $\cS_k(n)$. Thus, every $M$ consecutive rows in the matrix $\b{A}_n$ have the same $L-k$ columns with zeroes. That is, although at each slot at least $\min\{M,k\}-2$ independent vectors are added to $\b{A}_n$, at a certain slot, the vectors of $\b{A}'_n$ can be linearly dependent with those collected thus far in other $\b{A}'_j$ for $j<n$. This behavior is attributed to the "curse" of the coupon collector problem and is depicted in Figure \ref{fig-rank_A_as_func_n_M_relays}. Specifically, Figure \ref{fig-rank_A_as_func_n_M_relays} plots the average rank of $\b{A}_n$ as a function of $n$ for different values of $M$ and $L=150$. One can observe that the average rank of $\b{A}_n$ start growing linearly with the straight line $f(n)=Mn$. That is, at the beginning, in each slot $M$ DoF are added to $\b{A}_n$. However, at a certain point, it becomes harder to find new innovative coefficient vectors, and the curves start to flatten. Fortunately, for $M<<L$, the curves continue at the maximal rate almost up to a full rank. As $M$ grows, the curves flattens earlier.

Following the scheduling scheme suggested in Section \ref{Sec-Scheduling_in_CF}, the scheduler can thus randomly schedule users and collect the coefficient vectors up until recognizing a substantial decrease in the rate of the rank of $\b{A}_n$. At this point, the scheduler lets the users which their messages complete a full rank of $\b{A}_n$, transmit separately (similar to the final SISO step in Algorithm \ref{algo-scheduling algorithm for all transmission}). As Figure \ref{fig-rank_A_as_func_n_M_relays} depicts, the number of such slots is small when $M=O(1)$. Thus, the completion time of the transmission of all messages is roughly $N=\frac{L}{M-o(1)}=\frac{L}{M}+o(L)$ which results with a pre-log gain to the sum-rate. That is,
\begin{equation*}
\begin{aligned}
C_{SR}^{M}&=\frac{R(L-o(L))+R^{SISO}o(L)}{N}\\
&=M \frac{R(L-o(L))+R^{SISO}o(L)}{L+o(L)}\\
&\geq \frac{L-o(L)}{L+o(L)} M R= M R(1-o(1)),
\end{aligned}
\end{equation*}
Note that the transmission rate $R$ is set such that a relay will be able to decode successfully any linear combination when $L$ is large. That is, a rate that will satisfy zero outage probability at the limit of large $L$. Thus, the probability for outage as given in \eqref{equ-outage definition} should be modified to the following,

\small
\begin{equation}\label{equ-outage definition multiple relays}
P_{out}(R) = P_r\left(\min_{m:1,..,M}\left\{\min_{n:1,...,L}\{ \cR(\b{h}_m^*(n),\b{a}_m^*(n))\}\right\} < R \right).
\end{equation}
\normalsize
Where setting $R=(\frac{1}{4}-\epsilon)\log{\log{L}}$ for any small positive $\epsilon$ will satisfy 
\begin{equation}
\lim_{L \rightarrow \infty} P_{out}\left( R\right) =0.
\end{equation}

This can be proved if one follows the proof's steps of Lemma \ref{lem-outage probability goes to zero} while remembering that the values of $u$ and $\delta$ were chosen such that in each time slot a simultaneously good $k$ users for all the relays can be found with probability that goes to one with $L$. Specifically, one can use the lower bounds given in \eqref{equ-lower bound on h' and cos} in the achievable rate of each relay and thus remove the minimum on $m$.

\section{Conclusions}
The CF coding scheme provides a well understood framework for handling multiple transmissions, and decode them as linear combinations of messages using lattice codes. Accordingly, this enables a good utilization of the shared medium and can be employed in various communication systems. In this work, we have studied the impact of the number of transmitters in a CF system which was not addressed up until today. Specifically, when the receiver decodes linear combinations of \emph{all transmitted messages}, the number of transmitters heavily affects its ability to provide high computation rates. The analysis and results presented in the first part of this paper provide a good understanding on this effect and show that  as the number of transmitters grows, CF becomes degenerated, in the sense that the receiver prefers to decode a single message instead of a non-trivial linear combination. As a direct outcome, the computation rate and the system's sum-rate tend to zero since all other users are treated as noise. Thus, one is forced to restrict the number of transmitting users, i.e., use scheduling, in order to maintain the superior abilities CF provides.

In the second part of this work, we formulated the scheduling problem for a large scale CF system and presented a scheduling paradigm. The paradigm is based on the fact that one can always (with very high probability) find users that, if scheduled together, achieve the scaling law of the optimal computation rate. Moreover, the search of this simultaneously good users can be done in polynomial time considering the properties of what is considered as a good linear combination for CF. The analysis and results used in this work can be further applied to communication problems for which there is freedom of choosing the transmitters and thus the decoded linear combinations.


\appendices

\section{Proof for Theorem \ref{the-Achievable is going to zero}} 

\label{AppendixA} 

Define,
\begin{equation}\label{equ-Partition of channel vectors}
\begin{aligned}
& H_e=\{ \b{h}\in \R^L | \argmin_{\textbf{a}\in\mathbb{Z}^L \backslash \{\textbf{0}\}} f(\textbf{a})=\b{e}_i, \text{ for some } i \} \\
& H_{\overline{e}}=\{ \b{h}\in \R^L | \argmin_{\textbf{a}\in\mathbb{Z}^L \backslash \{\textbf{0}\}} f(\textbf{a}) \neq \b{e}_i, \text{ for all } i\}.
\end{aligned}
\end{equation}

That is, with probability $P_r(\b{e})$ a relay sees a channel vector $\b{h}\in H_e$ and with probability $P_r(\overline{\b{e}})$ a relay sees a channel vector $\b{h}\in H_{\overline{e}}$. Accordingly, the complementary CDF of the achievable rate can be expressed as, 
\begin{equation}\label{equ-The achievable rate partitioned to two terms}
P_r\left( \cR(\b{h},\b{a})>\epsilon \right)=P_r(\b{e}) P_r\left(  \cR(\b{h},\b{e}) >\epsilon |\b{h}\in H_{e}\right)+P_r(\overline{\b{e}}) P_r\left( \cR(\b{h},\b{a})>\epsilon | \b{h} \in H_{\overline{e}} \right).
\end{equation}

We treat the two terms above separately. The first term describes the case where the maximizing coefficient vector is some unit vector, while the second term describes the case where the maximizing coefficient vector may be any integer vector excluding the unit vectors. Note that we write $\b{e}$ to express the event that \emph{some} unit vector $\b{e}_i$ was chosen. We show that both terms tends to zero. Starting with the second term,

\small
\begin{equation}
\begin{aligned}
 P_r(\overline{\b{e}}) P_r\left( \cR(\b{h},\b{a})>\epsilon | \b{h} \in H_{\overline{e}} \right)&=P_r(\overline{\b{e}})P_r\left( \frac{1}{2}\log^+\left( \frac{1+\text{P}\|\b{h}\|^2}{ \|\b{a}\|^2+ \text{P}(\|\b{a}\|^2\|\b{h}\|^2 - (\b{h}^T\b{a})^2)}\right)>\epsilon \Big|\b{h} \in H_{\overline{e}} \right)\\
 &\leq P_r(\overline{\b{e}})P_r\left(  \frac{1}{2}\log^+\left( 1+\text{P}\|\b{h}\|^2 \right)>\epsilon \Big|\b{h} \in H_{\overline{e}}\right),
 \end{aligned}
\end{equation}
\normalsize
where the inequality is due to Cauchy-Schwarz and since $\sqn{a}>1$.

We upper bound the above using the Markov and Jensen's inequalities as follows, 
\begin{equation}\label{equ-Rate after Markov and Jensen's}
\begin{aligned}
P_r(\overline{\b{e}})P_r\left(  \frac{1}{2}\log^+\left( 1+\text{P}\|\b{h}\|^2 \right)>\epsilon \Big|\b{h} \in H_{\overline{e}}\right)&\leq P_r(\overline{\b{e}})\frac{1}{\epsilon} \mathop{\mathbb{E}} \left[\frac{1}{2}\log^+\left( 1+\text{P}\|\b{h}\|^2 \right) \Big|\b{h} \in H_{\overline{e}}  \right]\\
&\leq P_r(\overline{\b{e}})\frac{1}{2\epsilon}\log^+\left( 1+\text{P} \mathop{\mathbb{E}} \left[ \|\b{h}\|^2\big|\b{h} \in H_{\overline{e}} \right] \right).
\end{aligned}
\end{equation}
To further upper bound the above note that
\begin{equation}
\begin{aligned}
\E \left[ \|\b{h}\|^2\right]&=\E \left[ \|\b{h}\|^2\big|\b{h} \in H_{\overline{e}}\right]P_r(\overline{\b{e}})+\E \left[ \|\b{h}\|^2\big|\b{h} \in H_{e}\right]P_r(\b{e})\\
&\geq \E \left[ \|\b{h}\|^2\big|\b{h} \in H_{\overline{e}}\right]P_r(\overline{\b{e}}),
\end{aligned}
\end{equation}
which results with
\begin{equation}
\mathop{\mathbb{E}} \left[ \|\b{h}\|^2\big|\b{h} \in H_{\overline{e}}\right] \leq  \frac{L}{P_r(\overline{\b{e}})}
\end{equation}
since the channel vector $\b{h}$ is a Gaussian random vector and its squared norm follows the $\chi^2_L$ distribution.

Applying the expectation's upper bound in \eqref{equ-Rate after Markov and Jensen's} we have,
\begin{equation}\label{equ-Rate after Markov and Jensen's and upper bound on expectation}
\begin{aligned}
 P_r(\overline{\b{e}})\frac{1}{2\epsilon}\log^+\left( 1+\text{P} \mathop{\mathbb{E}} \left[ \|\b{h}^{\overline{e}}\|^2\big|\overline{\b{e}}\right] \right)&\leq P_r(\overline{\b{e}}) \frac{1}{2\epsilon}  \log^+\left( 1+\frac{\text{P}L}{P_r(\overline{\b{e}})}  \right)\\
&\overset{(a)}{\leq} P_r(\overline{\b{e}})\frac{1}{2\epsilon}  \sqrt{2 \frac{\text{P}L}{P_r(\overline{\b{e}})}}\\
&= \frac{1}{2\epsilon}  \sqrt{2\text{P}LP_r(\overline{\b{e}})}  \\
&\overset{(b)}{\leq} \frac{1}{2\epsilon}  \sqrt{8\text{P}^2L^3e^{-LE_2(L)}}\\
&= \frac{1}{\epsilon} \text{P}L\sqrt{2L}e^{-LE_3(L)},
\end{aligned}
\end{equation}  
where $(a)$ is due to $\log(1+x)\leq \sqrt{2x}$, $(b)$ follows from Theorem \ref{the-Probability for having a unit vector as the maximaizer over all other vectors} and $E_3(L)=\frac{1}{4}(1-\frac{1}{L})\log{2}$.
Considering the above, as $L$ grows, the second summand in \eqref{equ-The achievable rate partitioned to two terms} tends to zero for all $\epsilon>0$.

We are left with the first term in \eqref{equ-The achievable rate partitioned to two terms} which can be upper bounded as follows
\begin{equation}\label{equ-first term expression} 
\begin{aligned}
P_r(\b{e}) P_r\left(  \cR(\b{h},\b{e}) >\epsilon | \b{h}\in H_{e}\right) &\leq P_r\left(  \cR(\b{h},\b{e}) >\epsilon | \b{h}\in H_{e}\right)\\
& \overset{(a)}{=} P_r\left( \frac{1}{2} \log^+\left( \frac{1+\text{P}\|\b{h}\|^2}{1+\text{P}(\|\b{h}\|^2-h_i^2)} \right)>\epsilon \Big| \b{h}\in H_{e}\right)\\
& = P_r\left( \frac{1}{2} \log^+\left( \frac{1}{1- \frac{\text{P}h_i^2}{1+\text{P}\|\b{h}\|^2}} \right)>\epsilon  \Big| \b{h}\in H_{e}\right)\\
& \leq P_r\left( \frac{1}{2} \log^+\left( \frac{1}{1- \frac{h_i^2}{\|\b{h}\|^2}} \right)>\epsilon  \Big| \b{h}\in H_{e}\right)\\
& \overset{(b)}{=} P_r\left( \frac{h_i^2}{\|\b{h}\|^2} >1-\frac{1}{2^{2\epsilon}}  \Big| \b{h}\in H_{e}\right)\\
& \overset{(c)}{=} \frac{1}{\epsilon'} \E\left[\frac{h_i^2}{\|\b{h}\|^2} \Big| \b{h}\in H_{e}\right].
\end{aligned}
\end{equation}
where in $(a)$ we set the specific unit vector $\b{e}_i$ the relay has chosen for the channel vector $\b{h}\in H_{e}$. Note that the unit vector that maximizes the rate corresponds to the strongest transmitter, i.e., given $\b{h}$, the relay chooses $i=\argmax_i{\b{h}}$. In $(b)$, since the argument of the $\log$ is greater than one $\log^+(\cdot)= \log(\cdot)$. $(c)$ follows from Markov's inequality where we denote $\epsilon'=1-\frac{1}{2^{2\epsilon}}$. 

To further upper bound the above note that
\begin{equation}
\begin{aligned}
\E \left[ \frac{h_i^2}{\|\b{h}\|^2}\right]&=\E \left[ \frac{h_i^2}{\|\b{h}\|^2}\Big|\b{h} \in H_{\overline{e}}\right]P_r(\overline{\b{e}})\E \left[ \frac{h_i^2}{\|\b{h}\|^2}\Big|\b{h} \in H_{e}\right]P_r(\b{e})\\
&\geq \E \left[ \frac{h_i^2}{\|\b{h}\|^2}\Big|\b{h} \in H_{e}\right]P_r(\b{e}),
\end{aligned}
\end{equation}
which results with
\begin{equation}
\E \left[ \frac{h_i^2}{\|\b{h}\|^2}\Big|\b{h} \in H_{e}\right] \leq  \frac{1}{P_r(\b{e})}\E \left[ \frac{h_i^2}{\|\b{h}\|^2}\right].
\end{equation}
Theorem \ref{the-Probability for having a unit vector as the maximaizer over all other vectors} implies that $P_r(\b{e})\rightarrow 1$ as $L\rightarrow \infty$. In addition the numerator scales like $O(\ln{L})$ as was shown in the proof of Theorem \ref{the-Expected sum-rate of scheduling algorithm upper bound} and the denominator scales as $O(L)$ since $\|\b{h}\|^2$ follows the chi-squared distribution with $L$ DoF. Therefore, it is clear that as $L$ grows the first summand in \eqref{equ-The achievable rate partitioned to two terms} tends to zero for all $\epsilon'>0$. For an exact bound recall that 
\begin{equation*}
\E \left[ h_i^2\frac{1}{\|\b{h}\|^2}\right]=\E \left[ h_i^2\right]\E \left[\frac{1}{\|\b{h}\|^2}\right]+\text{cov}\left( h_i^2,\frac{1}{\|\b{h}\|^2}\right).
\end{equation*} 
Since the covariance is negative we have
\begin{equation}
\begin{aligned}
\E \left[ \frac{h_i^2}{\|\b{h}\|^2}\right] &\leq \EX  \left[(h_i)^2\right] \EX \left[ \frac{1}{\|\b{h}\|^2} \right]\\
&\leq \left(2\ln{L}+\frac{\gamma}{2}+o(1)\right)\frac{1}{L-2},
\end{aligned}
\end{equation}
where in the last line the expectations are for the maximum of a chi-squared r.v. \cite[Table 3.4.4]{embrechts2013modelling} and an inverse chi-squared with $L$ DoF r.v., respectively. As $L$ grows, the above tends to zero.

\section{Proof for Corollary \ref{the-Sum rate is going to zero}} 

\label{AppendixB} 

Considering Equation \eqref{equ-The achievable rate partitioned to two terms} in Appendix \ref{AppendixA}, we can upper bound $P_r(C_{SR}>\epsilon)$ as follows,
\begin{equation}\label{equ-The sum rate partitioned to two terms}
\begin{aligned}
 P_r(C_{SR}>\epsilon)&\triangleq P_r\left(\frac{R  L}{N}>\epsilon\right)\\
 &\leq P_r\left(\frac{L}{N} \min_{n=1,...,N}\min_{m} \cR(\b{h}_m(n),\b{a}_m(n))>\epsilon\right)\\
 &\leq P_r\left(\frac{L}{N} \cR(\b{h}_m(n),\b{a}_m(n))>\epsilon\right)\\
 &= P_r(\b{e}) P_r\left( \frac{L}{N} \cR(\b{h}_m,\b{e}) >\epsilon | \b{h}\in H_{e}\right)+P_r(\overline{\b{e}}) P_r\left(\frac{L}{N} \cR(\b{h}_m,\b{a}_m)>\epsilon |\b{h} \in H_{\overline{e}} \right).
\end{aligned}
\end{equation}

Following the same steps as the proof for Theorem \ref{the-Achievable is going to zero} and remembering that $\frac{L}{N}$ can be upper bounded by $M$, we have 

\begin{equation}\label{equ-second term limit sum-rate}
\begin{aligned}
&\lim_{L \rightarrow \infty}  P_r(\overline{\b{e}}) P_r\left(\frac{L}{N} \cR(\b{h}_m,\b{a}_m)>\epsilon |\b{h} \in H_{\overline{e}} \right)\leq \lim_{L \rightarrow \infty} \frac{1}{\epsilon} \frac{\text{P}ML\sqrt{L}}{\sqrt{0.5}}e^{-LE_3(L)} =0
\end{aligned}
\end{equation}
for all $\epsilon>0$.
Similarly, the first term is upper bounded by

\begin{equation}
\begin{aligned}
 \lim_{L \rightarrow \infty} P_r(\b{e}) P_r\left( \frac{L}{N} \cR(\b{h}_m,\b{e}) >\epsilon | \b{h} \in H_{e}\right)&\leq \lim_{L \rightarrow \infty} P_r\left(\frac{1}{2} M \log^+\left( \frac{1}{1- \frac{(h_i)^2}{\|\b{h}\|^2}} \right)>\epsilon |\b{h} \in H_{e}\right)\\
 &=0,
 \end{aligned}
\end{equation}
provided that $M$ is fixed.

\section{Proof for lemma \ref{lem-outage probability goes to zero}}
\label{AppendixH}
\begin{align*}
\lim_{L \rightarrow \infty} P_{out}\left( R\right)
&=\lim_{L \rightarrow \infty} P_r\left(\min_{n}\{ \cR(\b{h}^*(n),\b{a}^*(n))\} < \left(\frac{1}{4}-\epsilon\right)\log{\log{L}} \right)\\
&=\lim_{L \rightarrow \infty} P_r\left(\min_{n}\left\{  \frac{1}{2} \log^+ \left( k- \frac{\text{P}({\b{h}^*(n)}^T\b{1})^2}{1+\text{P}\|\b{h}^*(n)\|^2} \right)^{-1} \right\}  < \left(\frac{1}{4}-\epsilon\right)\log{\log{L}} \right) \\
&\overset{(a)}{\leq} \lim_{L \rightarrow \infty} P_r\left(  \frac{1}{2} \log \left( k- \min_{n}\left\{ \frac{\text{P}({\b{h}^*(n)}^T\b{1})^2}{1+\text{P}\|\b{h}^*(n)\|^2}\right\} \right)^{-1}   < \left(\frac{1}{4}-\epsilon\right)\log{\log{L}} \right) \\
&=\lim_{L \rightarrow \infty} P_r\left( \left( k- \min_{n}\left\{ \frac{\text{P}({\b{h}^*(n)}^T\b{1})^2}{1+\text{P}\|\b{h}^*(n)\|^2}\right\} \right)  > (\log{L})^{-\left(\frac{1}{2}-2\epsilon\right)} \right) \\
& \overset{(b)}{\leq}  \lim_{L \rightarrow \infty} \frac{\E\left[ k- \min_{n}\left\{ \frac{\text{P}({\b{h}^*(n)}^T\b{1})^2}{1+\text{P}\|\b{h}^*(n)\|^2}\right\} \right]}{(\log{L})^{-\left(\frac{1}{2}-2\epsilon\right)}}\\
& = \lim_{L \rightarrow \infty} k (\log{L})^{\left(\frac{1}{2}-2\epsilon\right)}\left( 1- \E\left[\min_{n}\left\{ \frac{\text{P}\sqn{\b{h}^*(n)}\cos^2(\theta(n))}{1+\text{P}\|\b{h}^*(n)\|^2}\right\} \right]\right)\\
& \leq \lim_{L \rightarrow \infty} k (\log{L})^{\left(\frac{1}{2}-2\epsilon\right)}\left( 1- \E\left[\min_{n}\left\{ \frac{\cos^2(\theta(n))}{\frac{1}{\text{P}\sqn{\b{h}^*(n)}}+1}\right\} \Bigg| \xi \right]P_r(\xi)\right)\\
& \overset{(c)}{\leq}\lim_{L \rightarrow \infty}  k (\log{L})^{\left(\frac{1}{2}-2\epsilon\right)} \left( 1- \E\left[\min_{n}\left\{ \frac{\frac{u^2}{(u+\delta)^2}}{\frac{1}{\text{P}ku^2}+1}\right\} \Bigg| \xi \right]P_r(\xi)\right)\\
& = \lim_{L \rightarrow \infty} k (\log{L})^{\left(\frac{1}{2}-2\epsilon\right)} \left( 1- \frac{\frac{u^2}{(u+\delta)^2}}{\frac{1}{\text{P}ku^2}+1}P_r(\xi)\right)\\
& \leq \lim_{L \rightarrow \infty} k (\log{L})^{\left(\frac{1}{2}-2\epsilon\right)} \left( \frac{1+ \text{P}ku^2 -\frac{u^2P_r(\xi)}{(u+\delta)^2}\text{P}ku^2}{\text{P}ku^2}\right)\\
& = \lim_{L \rightarrow \infty}  k(\log{L})^{\left(\frac{1}{2}-2\epsilon\right)} \left( \frac{1} {\text{P}ku^2} +1 -\frac{u^2P_r(\xi)}{(u+\delta)^2}\right)\\
& =\lim_{L \rightarrow \infty}  \frac{(\log{L})^{\left(\frac{1}{2}-2\epsilon\right)}}{\text{P}u^2}+ \lim_{L \rightarrow \infty}  k(\log{L})^{\left(\frac{1}{2}-2\epsilon\right)} \left( 1-\frac{u^2P_r(\xi)}{(u+\delta)^2}\right)\\
& \overset{(d)}{=} \lim_{L \rightarrow \infty} k(\log{L})^{\left(\frac{1}{2}-2\epsilon\right)} \left( 1-\frac{u^2P_r(\xi)}{(u+\delta)^2}\right)\\
& = \lim_{L \rightarrow \infty}   \frac{k(\log{L})^{\left(\frac{1}{2}-2\epsilon\right)}}{(u+\delta)^2}\left( u^2(1-P_r(\xi))+2u\delta+\delta^2 \right)\\
& \overset{(e)}{=} \lim_{L \rightarrow \infty}   \frac{k(\log{L})^{\left(\frac{1}{2}-2\epsilon\right)}}{(u+\delta)^2} \cdot \lim_{L \rightarrow \infty}\left( u^2(1-P_r(\xi)) \right)\\
&  \overset{(f)}{\leq}  \lim_{L \rightarrow \infty}   \frac{k(\log{L})^{\left(\frac{1}{2}-2\epsilon\right)}}{(u+\delta)^2} \lim_{L \rightarrow \infty}\left( u^2\left(1-\left(1-Le^{-\sqrt{L} +(\ln{(L)}+1)} \right)\right) \right)\\
& \overset{(g)}{\leq} 0 \cdot \lim_{L \rightarrow \infty}\left( u^2Le^{-\sqrt{L} +(\ln{(L)}+1)}\right)\\
&=0.
\end{align*}
$(a)$ follows since $\log^+(x)\geq\log(x)$. In $(b)$ we used the Markov inequality and in $(c)$ we bound the norm and the cosine as was done in the proof of Theorem \ref{the-Expected achievable rate of scheduling algorithm lower bound}. Note that the values of $u$ and $\delta$ were chosen to fit all slots and in particular to the slot with the minimal achievable rate. $(d)$, $(e)$ and $(g)$ follows since $u=O(\sqrt{\log{L}})$ and $\delta=O(\frac{1}{\log{L}})$. In $(f)$ we used \eqref{equ-limit of P_r(xi) goes to one}.

\section{Proof for the scaling laws of Theorem \ref{the-Expected achievable rate of scheduling algorithm lower bound}}

\label{AppendixC}
To prove that the scaling law is $\frac{1}{4}\log{\log{L}}$, we show that the limit of the division of the lower bound with $\frac{1}{4}\log{\log{L}}$ equals 1 as follows,
\small
\begin{equation}
\begin{aligned}
	\lim_{L\rightarrow \infty} \frac{\frac{1}{2}  \log^+ \left( k \left(1-\frac{\text{P}ku^4}{(u+\delta)^2(1+\text{P}ku^2)}\left(1-g(L)\right)^L \right)\right)^{-1}}{\frac{1}{4}\log{\log{L}}}
	= \lim_{L\rightarrow \infty} \frac{ -2 \log^+ \left( k \left(1-\frac{\text{P}ku^4\left(1-g(L)\right)^L}{(u+\delta)^2(1+\text{P}ku^2)} \right)\right)}{\log{\log{L}}}\\
\end{aligned}
\end{equation}
where $g(L)=e^{-\frac{1}{2p(u,\delta)}\frac{(Lp(u,\delta)-(k-1))^2}{L}}$ which we expressed as $o(1)$ in the theorem. We start with an upper bound on the ratio.
\begin{align*}
	\lim_{L\rightarrow \infty} \frac{ -2 \log^+ \left( k \left(1-\frac{\text{P}ku^4\left(1-g(L)\right)^L}{(u+\delta)^2(1+\text{P}ku^2)} \right)\right)}{\log{\log{L}}}
	&\leq \lim_{L\rightarrow \infty} \frac{ -2 \log^+ \left( k \left(1-\frac{u^2}{(u+\delta)^2} \right)\right)}{\log{\log{L}}}\\
	&\leq \lim_{L\rightarrow \infty} \frac{ -2 \log^+ \left( k \left(\frac{\delta}{(u+\delta)} \right)\right)}{\log{\log{L}}}\\
	&\overset{(a)}{=} \lim_{L\rightarrow \infty} \frac{ -2 \log^+ \left( \frac{3}{(u+\delta)}\right)}{\log{\log{L}}}\\
	&\overset{(b)}{=} \lim_{L\rightarrow \infty} \frac{ 2 \log^+ \left( \sqrt{2\ln{\frac{ \delta\sqrt{L}}{\sqrt{2\pi}}}} \right)}{\log{\log{L}}}\\
	&= \lim_{L\rightarrow \infty} \frac{ \log^+ \left(2\ln{\delta}+\ln{L}-\ln{2\pi}  \right)}{\log{\log{L}}}\\
	&\leq \lim_{L\rightarrow \infty} \frac{ \log^+ \left(\ln{L}  \right)}{\log{\log{L}}} =1.\\
\end{align*}
In $(a)$ we used our choice of $\delta=\frac{1}{\ln{L}}$ and $k=\lceil\ln{L}\rceil+1$ to upper bound $k\delta<3$. In $(b)$ we set $u+\delta=\sqrt{2\ln\frac{\delta\sqrt{L}}{\sqrt{2 \pi}}}$. The lower bound is as follows,
\begin{align*}
	 &\lim_{L\rightarrow \infty} \frac{ -2 \log^+ \left( k \left(1-\frac{\text{P}ku^4\left(1-g(L)\right)^L}{(u+\delta)^2(1+\text{P}ku^2)} \right)\right)}{\log{\log{L}}}\\
	& \geq \lim_{L\rightarrow \infty} \frac{ -2 \log^+ \left( k \left(1-\frac{ku^4\left(1-g(L)\right)^L}{(u+\delta)^2(1+ku^2)} \right)\right)}{\log{\log{L}}}\\
	&= \lim_{L\rightarrow \infty} \frac{ -2 \log^+ \left( k \left(1-\frac{ku^4\left(1-g(L)\right)^L}{u^2+2u\delta+\delta^2+ku^4+2ku^3\delta+ku^2\delta^2} \right)\right)}{\log{\log{L}}}\\
	&  \overset{(c)}{\geq} \lim_{L\rightarrow \infty} \frac{ -2 \log^+ \left( k \left(1-\frac{ku^4\left(1-g(L)\right)^L}{u^2+2u^2\delta k+u^2\delta^2k+ku^4+2ku^3\delta+ku^2\delta^2} \right)\right)}{\log{\log{L}}}\\
	& = \lim_{L\rightarrow \infty} \frac{ -2 \log^+ \left( k \left(1-\frac{ku^2\left(1-g(L)\right)^L}{1+2k\delta(1+\delta+u)+ku^2} \right)\right)}{\log{\log{L}}}\\
	& = \lim_{L\rightarrow \infty} \frac{ -2 \log^+ \left( k \left(\frac{1+2k\delta(1+\delta+u)+ku^2\left(1-\left(1-g(L)\right)^L\right)}{1+2k\delta(1+\delta+u)+ku^2} \right)\right)}{\log{\log{L}}}\\
	& \overset{(d)}{\geq} \lim_{L\rightarrow \infty} \frac{ -2 \log^+ \left( \frac{7(1+\delta+u)+ku^2\left(1-\left(1-g(L)\right)^L\right)}{u^2} \right)}{\log{\log{L}}}\\
	& = \lim_{L\rightarrow \infty} \frac{ 2 \log^+ \left(u^2 \right)}{\log{\log{L}}} -\lim_{L\rightarrow \infty} \frac{ 2 \log^+ \left(7(1+\delta+u)+ku^2\left(1-\left(1-g(L)\right)^L\right) \right)}{\log{\log{L}}}  \\
	& \overset{(e)}{=} 2 -\lim_{L\rightarrow \infty} \frac{ 2 \log^+ \left(7(1+\delta+u)+ku^2\left(1-\left(1-g(L)\right)^L\right) \right)}{\log{\log{L}}}  \\
	& = 2 -\lim_{L\rightarrow \infty} \frac{ 2 \log^+ \left(7\left(1+\delta+u \right)\right)}{\log{\log{L}}}  - \lim_{L\rightarrow \infty} \frac{ 2 \log^+ \left(1+\frac{ku^2}{7(1+\delta+u)}\left(1-\left(1-g(L)\right)^L\right) \right)}{\log{\log{L}}}\\
	&\overset{(f)}{\geq} 2-1 - \lim_{L\rightarrow \infty} \frac{ \frac{2ku^2}{7(1+\delta+u)}\left(1-\left(1-g(L)\right)^L \right)}{\log{\log{L}}}\\
	&\geq 1 - \lim_{L\rightarrow \infty} \frac{ \frac{2k(u+\delta)^2}{7(\delta+u)}\left(1-\left(1-g(L)\right)^L \right)}{\log{\log{L}}}\\
	&\geq 1 - \lim_{L\rightarrow \infty} \frac{ k(u+\delta)\left(1-\left(1-g(L)\right)^L \right)}{\log{\log{L}}}\\
	&\geq 1 - \lim_{L\rightarrow \infty} \frac{ 4(\ln{L})^2\left(1-\left(1-g(L)\right)^L \right)}{\log{\log{L}}}\\
	&\geq 1 - \lim_{L\rightarrow \infty}  4(\ln{L})^2\left(1-\left(1-g(L)\right)^L \right)\\
	&\overset{(g)}{\geq} 1 - \lim_{L\rightarrow \infty}  4(\ln{L})^2\left(1-\left(1-Le^{-\sqrt{L} +(\ln{(L)}+1)}\right) \right)\\
        &= 1 - \lim_{L\rightarrow \infty}  4(\ln{L})^2Le^{-\sqrt{L} +(\ln{(L)}+1)}\\
        &=1-0=1
	\end{align*}
\normalsize
$(c)$ follows since, for large enough $L$, $u^2k>1$ and $uk>1$. In $(d)$, we used again the bound $k\delta<3$, the fact that $\delta>0$ and that for large enough $L$, $u>0$. $(e)$ follows from the following,
\begin{align*}
 \lim_{L\rightarrow \infty} \frac{ 2 \log^+ \left(u^2 \right)}{\log{\log{L}}}
 &=\lim_{L\rightarrow \infty} \frac{ 4 \log^+ \left(\sqrt{2\ln\frac{\delta\sqrt{L}}{\sqrt{2 \pi}}} -\delta \right)}{\log{\log{L}}}\\
 &=\lim_{L\rightarrow \infty} \frac{ 4 \log^+ \left(\sqrt{ \ln L +2\ln \delta -\ln 2 \pi} -\delta \right)}{\log{\log{L}}}\\
 &=2.
\end{align*}
In $(f)$, the third term can be upper bounded using the relation of $\log{(1+x)\leq x}$ for $x>0$ and the second term follows from the following,
\begin{align*}
\lim_{L\rightarrow \infty} \frac{ 2 \log^+ \left(7\left(1+\delta+u \right)\right)}{\log{\log{L}}} 
&=\lim_{L\rightarrow \infty} \frac{ 2 \log^+ \left(1+\delta+u \right)}{\log{\log{L}}} \\
&=\lim_{L\rightarrow \infty} \frac{ 2 \log^+ \left(1+\sqrt{2\ln\frac{\delta\sqrt{L}}{\sqrt{2 \pi}}} \right)}{\log{\log{L}}} \\
&=1.
\end{align*}
Finally, in $(g)$ we used equation \eqref{equ-limit of P_r(xi) goes to one}. The correctness of these limits was verified also on Mathematica.

\section{Proof for the scaling laws of Theorem \ref{the-Expected sum-rate of scheduling algorithm upper bound} } 

\label{AppendixD}

To prove that the scaling law is $\frac{1}{2}\log{\log{L}}$, we show that the limit of the division of the lower bound with $\frac{1}{2}\log{\log{L}}$ equals 1 as follows,
\small
\begin{equation}
\begin{aligned}
	&\lim_{L\rightarrow \infty} \frac{\frac{1}{2}\log{\left(1+\text{P}\left(2\ln{L}-\ln{\ln{L}}-2\ln{\Gamma\left(\frac{1}{2}\right)}+\frac{\gamma}{2}\right)\right)}}{\frac{1}{2}\log{\log{L}}}.
\end{aligned}
\end{equation}
We start with an upper bound on the ratio.
\begin{equation}
\begin{aligned}
	\lim_{L\rightarrow \infty} \frac{\log{\left(1+\text{P}\left(2\ln{L}-\ln{\ln{L}}-2\ln{\Gamma\left(\frac{1}{2}\right)}+\frac{\gamma}{2}\right)\right)}}{\log{\log{L}}}
	&\leq \lim_{L\rightarrow \infty} \frac{\log{\left(\text{P}\left(1+2\ln{L}-\ln{\ln{L}}-2\ln{\Gamma\left(\frac{1}{2}\right)}+\frac{\gamma}{2}\right)\right)}}{\log{\log{L}}}\\
	&\leq \lim_{L\rightarrow \infty} \frac{\log{\left(1+2\ln{L}+\frac{\gamma}{2}\right)}}{\log{\log{L}}}=1.
\end{aligned}
\end{equation}
The lower bound is as follows,
\begin{equation}
\begin{aligned}
	\lim_{L\rightarrow \infty} \frac{\log{\left(1+\text{P}\left(2\ln{L}-\ln{\ln{L}}-2\ln{\Gamma\left(\frac{1}{2}\right)}+\frac{\gamma}{2}\right)\right)}}{\log{\log{L}}}
	&\geq \lim_{L\rightarrow \infty} \frac{\log{\left(2\ln{L}-\ln{\ln{L}}-2\ln{\Gamma\left(\frac{1}{2}\right)}+\frac{\gamma}{2}\right)}}{\log{\log{L}}}\\
	&\geq \lim_{L\rightarrow \infty} \frac{\log{\left(\ln{\ln{L}}-2\ln{\Gamma\left(\frac{1}{2}\right)}\right)}}{\log{\log{L}}}=1.
\end{aligned}
\end{equation}

\section*{Acknowledgment}

The authors would like to thank Or Ordentlich for his contribution in Lemma \ref{lem-optimal schedule for all one vector}.

\bibliographystyle{IEEEtran}
\bibliography{Bibliography}

%








\end{document}